\documentclass[journal]{IEEEtran}
\usepackage{amsfonts}
\usepackage{color}
\usepackage{amsmath,amsfonts,amssymb,amsthm,epsfig,epstopdf,url,array}
\usepackage{url,textcomp}
\usepackage{authblk}
\usepackage{cite}
\usepackage{algorithm}
\usepackage{algpseudocode}
\usepackage{caption,subcaption}
\usepackage{multirow}
\usepackage{diagbox}
\DeclareMathAlphabet{\mathpzc}{OT1}{pzc}{m}{it}
\usepackage{dsfont}
\makeatletter
\newcommand*{\rom}[1]{\expandafter\@slowromancap\romannumeral #1@}
\makeatother

\newtheorem{theorem}{Theorem}

\begin{document}
\title{Effective Capacity for Renewal Service Processes with Applications to HARQ Systems}
\author{Zheng~Shi,
    Theodoros Tsiftsis,
    Weiqiang~Tan,
    Guanghua~Yang,
    Shaodan~Ma,
    and Mohamed-Slim Alouini
\thanks{Zheng Shi, Theodoros Tsiftsis and Guanghua Yang are with the School of Electrical and Information Engineering, Jinan University, Zhuhai 519070, China (e-mails: shizheng0124@gmail.com, theodoros.tsiftsis@gmail.com, ghyang@jnu.edu.cn).  
}
\thanks{Weiqiang~Tan is with the School of Computer Science and Educational Software, Guangzhou University, Guangzhou 510006, China (e-mail:wqtan@gzhu.edu.cn).
}
\thanks{Shaodan~Ma is with the Department of Electrical and Computer Engineering, University of Macau, Macao S.A.R., China (e-mail: shaodanma@umac.mo).
}
\thanks{Mohamed-Slim Alouini is with CEMSE Division, King Abdullah University of Science and Technology, Thuwal 23955-6900, Saudi Arabia (e-mail:slim.alouini@kaust.edu.sa).
}
}
\maketitle
\begin{abstract}
Considering the widespread use of effective capacity in cross-layer design and the extensive existence of renewal service processes in communication networks, this paper thoroughly investigates the effective capacity for renewal processes. Based on Z-transform, we derive exact analytical expressions for the effective capacity at a given quality of service (QoS) exponent for both the renewal processes with constant reward and with variable rewards. Unlike prior literature that the effective capacity is approximated with no many insightful discussions, our expression is simple and reveals further meaningful results, such as the monotonicity and bounds of effective capacity. The analytical results are then applied to evaluate the cross-layer throughput for diverse hybrid automatic repeat request (HARQ) systems, including fixed-rate HARQ (FR-HARQ, e.g., Type I HARQ, HARQ with chase combining (HARQ-CC) and HARQ with incremental redundancy (HARQ-IR)), variable-rate HARQ (VR-HARQ) and cross-packet HARQ (XP-HARQ). 
Numerical results corroborate the analytical ones and prove the superiority of our proposed approach. Furthermore, targeting at maximizing the effective capacity via the optimal rate selection, it is revealed that VR-HARQ and XP-HARQ attain almost the same performance, and both of them perform better than FR-HARQ.
\end{abstract}

\IEEEpeerreviewmaketitle
\section{Introduction}\label{sec:sys_mod}

\IEEEPARstart{5}{G} has been envisioned to offer extremely high data rate (on the order of Gbps), ultra-reliability (higher than 99.999\%) and very low latency (sub-1ms)\cite{popovski2014ultra}. 
To fulfill these demanding requirements, plenty of existing works prefer to evaluate and devise communication systems from physical-layer perspective. However, the constraints of link-layer quality of service (QoS) were rarely considered in the literature, such as queue length limitation and maximum allowable delay \cite{wu2003effective}. Unfortunately, physical-layer models can not capture the characteristics of these QoS requirements, which depends on the queueing behavior of the connection model. Ignoring these QoS limitations grossly overestimates the network performance \cite{to2015power}, and many QoS-constrained applications cannot be supported \cite{phan2016optimal}. Therefore, it is imperative to come up with a cross-layer performance metric that combines physical-layer parameters and QoS requirements together.

To address the aforementioned issue, the concept of effective capacity was developed in \cite{wu2003effective} initially by considering the finite buffer size in practice. This concept has been extensively employed to evaluate the maximum supportable arrival rate given a QoS exponent, where the QoS exponent affects the statistical QoS constraints, including buffer overflow probability and delay-violation probability. Furthermore, the effective capacity enables us to optimally design cross-layer parameters for various wireless systems subject to statistical QoS constraints \cite{abrao2017achieving,hu2018optimal}. However, most of the relevant literature assumed perfect knowledge of channel state information (CSI) at the transmitter, which is an impractical assumption due to the unpredictable noise, quantization errors, etc. 
Particularly, in the absence of perfect CSI, the retransmission technique of hybrid automatic repeat request (HARQ) is frequently utilized to enhance the system reliability. Nonetheless, the introduction of HARQ would make the queueing behaviour of the connection vastly involved, and consequently impedes the analysis of effective capacity given diverse QoS requirements. Instead, in \cite{choi2012large}, large deviation was adopted to convert the physical-layer throughput of HARQ-IR into effective capacity approximately. Moreover, the concept of effective throughput was developed to bypass the complex effective capacity in \cite{to2015power}. In \cite{li2015throughput}, an accurate approximation of effective capacity under small QoS exponent was obtained on the basis of the cumulants of renewal processes. Whereas, the concept of the effective capacity in \cite{li2015throughput} represents the maximum arrival rate that can be supported by HARQ systems. No matter whether the conveyed message can be recovered by receiver or not, every HARQ cycle will be counted as a success. Obviously, it will overestimate the link-layer throughput particularly for high probability of the decoding failures. Therefore, only the goodput of HARQ systems was considered into the formulation of the effective capacity in \cite{larsson2016effective}, and the effective capacity was obtained by using the recurrence relation approach. The similar results were further extended to examine the outage effective capacity of the buffer-aided diamond relay systems in \cite{qiao2016outage}. However in both \cite{larsson2016effective,qiao2016outage}, the analytical results are only applicable to fixed-rate HARQ (FR-HARQ) schemes, wherein the transmission rates remain constant during all HARQ rounds.

Unfortunately so far, the effective capacities of more advanced and complex HARQ schemes, e.g., variable-rate HARQ (VR-HARQ) \cite{szczecinski2013rate} and cross-packet HARQ (XP-HARQ) \cite{jabi2017adaptive}, have never been investigated due to the challenge of analyzing more complicated service process. Moreover, even if the closed-form expressions have been derived for the effective capacity of the conventional HARQ schemes \cite{larsson2016effective,qiao2016outage}, the complex expressions of the effective capacity provided little insights and it is also difficult to extend the analytical results to the general case. Hence, they will impede the effective cross-layer design of HARQ systems to further enhance the system performance. To combat this issue and generalize the analytical results, we notice that the HARQ transmission model can actually be described by a renewal reward process \cite{caire2001throughput}. Specifically, the event that the transmitter halts HARQ transmissions for the current message is recognized to be a renewal, and the number of the transmitted information bits reflects the reward received from the renewal. In this paper, we first derive a simple exact expression for the effective capacity of the network services that follows constant reward renewal process. The results are further extended to the general renewal process with variable rewards. The simple analytical expressions not only offer accurate approximation for the effective capacity under small QoS exponent, but also facilitate the extraction of further meaningful insights. In particular, the effective capacity decreases with the QoS exponent and is bounded. The analytical results are then applied to calculate link-layer throughputs for different HARQ systems, including the conventional FR-HARQ (e.g., Type I HARQ, HARQ with chase combining (HARQ-CC), HARQ with incremental redundancy (HARQ-IR)), VR-HARQ and XP-HARQ. Numerical examples are finally presented to confirm the proposed approach compared with the already existing ones. 
Furthermore, aiming to maximize the effective capacity through the optimal rate selection, the numerical results reveal that XP-HARQ and VR-HARQ reach almost the same performance in terms of the optimal effective capacity. 



The remainder of this paper is structured as follows. Section \ref{sec:pre} presents preliminaries on effective capacity and HARQ schemes. In Section \ref{sec:ana}, the effective capacity for the renewal process with constant reward is derived by means of Laplace transform and Z-transform, respectively. Section \ref{sec:ec_gen} then extends the results to the general renewal reward process. The analytical results are further applied to evaluate the link-layer throughput of various HARQ systems in Section \ref{sec:appli}. Numerical results are presented for verifications and discussions in Section \ref{sec:ver}. Section \ref{sec:con} finally concludes this paper.
\section{Preliminaries}\label{sec:pre}
\begin{figure*}[!t]
  \centering
  \includegraphics[height=2.3in]{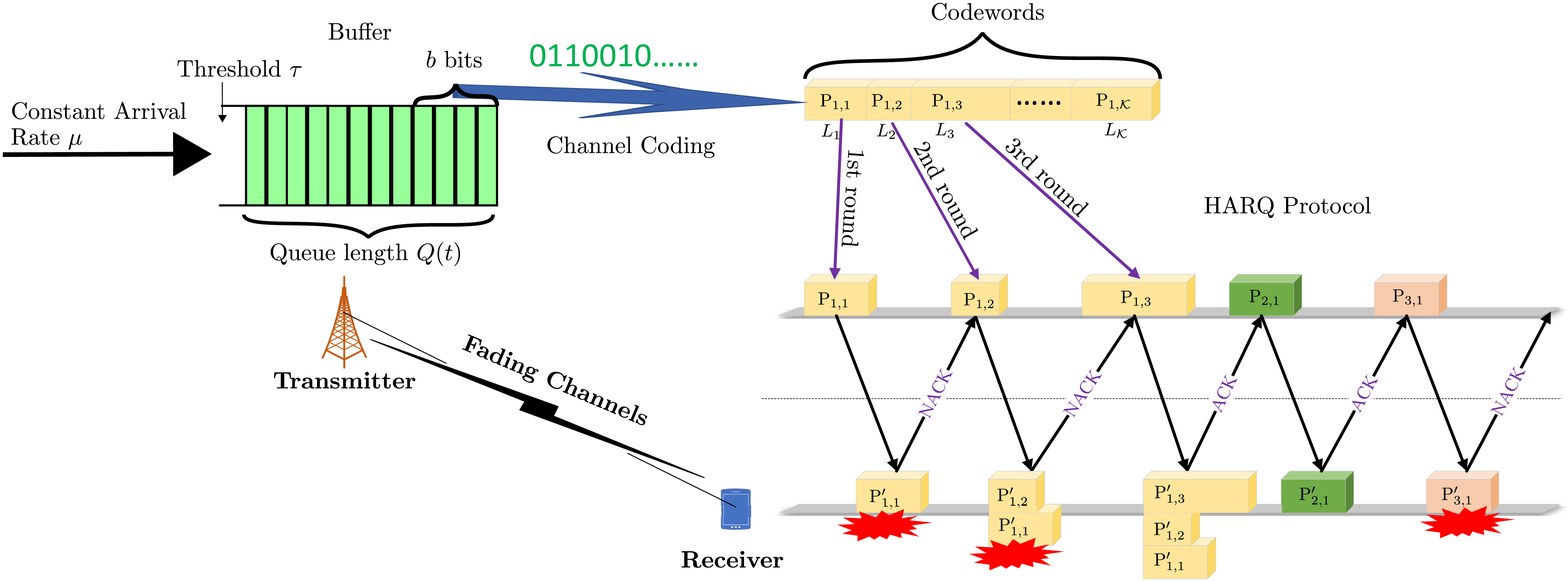}
  \caption{Cross-layer model for the buffer-limited HARQ System.}\label{fig:sys_mod}
\end{figure*}
\subsection{Effective Capacity}
The cross-layer model of the typical HARQ system with limited buffer is shown in Fig. \ref{fig:sys_mod}. By assuming a constant arrival rate $\mu$, buffer overflows will happen if the queue length $Q$ exceeds the buffer threshold $\tau$, and it can not be avoided under limited buffer size. Based on the theory of large deviations, the buffer overflow probability can be approximated as ${\Pr \left( {Q \ge \tau } \right)} \approx e^{-\theta \tau}$ for large values of $\tau$ \cite{chang1995effective}, where $Q$ and $\tau$ stand for the stationary queue length and buffer's threshold, respectively. It is worth noting that the QoS exponent $\theta$ plays a critical role in connecting the physical and link layers. The QoS exponent $\theta$ implies how fast the buffer overflow probability ${\Pr \left( {Q \ge \tau } \right)}$ decays exponentially with $\tau$. Accordingly, once a constraint is imposed on the buffer overflow probability, the QoS exponent $\theta$ can be determined as $\theta  \approx {{-\ln \left( {\Pr \left( {Q \ge \tau } \right)} \right)}}/{\tau }$. To meet the requirement of the QoS exponent, the constant arrival rate $\mu$ should be properly chosen. To this end, the concept of the effective capacity was developed to specify the maximum supportable arrival rate, i.e., $C_e = \max \{\mu\}$ \cite{wu2003effective}. Assume that the service process satisfies the G\"artner-Ellis theorem \cite{chang1995effective}, given the QoS exponent $\theta$, the effective capacity is explicitly given by the limit \cite{wu2003effective}
\begin{equation}\label{eqn:mod_ec}
{C_e} =  - \mathop {\lim }\limits_{t \to \infty } \frac{1}{{\theta t}}\ln \mathbb E\left\{ {{e^{ - \theta S_t}}} \right\},
\end{equation}
where $\mathbb E\{\cdot\}$ is the expectation operator and $S_t$ is the time-accumulated service process representing the total amount of reward received until time $t$. Moreover, the definition of the effective capacity was further expanded to the case with finite time $t$ in \cite{larsson2016effective}, i.e., ${C_{e,t}} =  - \frac{1}{{\theta t}}\ln \mathbb E\left\{ {{e^{ - \theta S_t}}} \right\}$.
\subsection{HARQ Schemes}
This paper offers a through analysis of the effective capacity for various HARQ systems. As a widely used reliable transmission technique, HARQ leverages both the forward error correction coding and automatic repeat request such that the currently received packet could be combined with the erroneously received packets to diminish the probability of decoding failures. As shown in Fig. \ref{fig:sys_mod}, according to the HARQ mechanism, the delivered message is first encoded into a long codeword, and the generated codeword is then broken into multiple subcodewords. These subcodewords will be sent sequentially in different HARQ rounds upon the reception of negative acknowledgment (NACK) messages, while the feedback of the positive acknowledgment (ACK) messages notifies the transmitter about the initiation of a new HARQ process for the next information message. On the basis of whether all the subcodeword lengths are fixed or not, the HARQ scheme can be classified into FR-HARQ and VR-HARQ. According to different encoding/decoding operations performed at the transceiver, the conventional FR-HARQ can be further categorized into three types, i.e., Type I HARQ, HARQ-CC and HARQ-IR \cite{caire2001throughput}. To be specific, Type I HARQ does not require the aid of the buffer at the receiver because the failed packets are directly discarded. Whereas, both HARQ-CC and HARQ-IR have the dedicated buffer to store the erroneously received packets, and their major difference lies in that diversity combining and code combining are employed for the decodings of HARQ-CC and HARQ-IR, respectively. In contrast with the conventional FR-HARQ, VR-HARQ assumes the variable lengths of the subcodewords for further throughput enhancement \cite{szczecinski2013rate}. It is worth noting that both FR- and VR-HARQ do not include new information bits into retransmissions, they may yield the waste of mutual information. To overcome this issue, VR-HARQ is proposed to add new information bits to retransmissions for possible redundant mutual information \cite{jabi2017adaptive}.

Note that the cross-layer HARQ system can be modelled by using the renewal reward process\cite{caire2001throughput}. Specifically, the event that HARQ retransmissions stop for the current message is treated as a renewal, and each time a renewal takes place we receive a reward, which is the total amount of the delivered information bits. Most of the network services obey renewal reward processes, which incorporate the HARQ channel service as a special case. Therefore, to ease extension, the effective capacity of the renewal service process with constant reward is derived first in Section \ref{sec:ana}. A unified closed-form expression of the effective capacity is then obtained for the complex renewal service process with variable rewards in Section \ref{sec:ec_gen}. The analytical results are then employed to obtain the effective capacity of various HARQ schemes. For the simplification of the analysis, the receptions of the HARQ feedback messages are assumed to be error- and delay-free.

\subsection{Two Useful Theorems}
To start with the analysis, two important theorems will be repeatedly used in the sequel. The first theorem is utilized to derive the effective capacity in a matrix form.
\begin{theorem}\label{the:inv_z_mat}
If $z_1,\cdots,z_K$ are distinct, the zero points of the function $\phi(z)$ does not coincide with $z_1,\cdots,z_K$, and $\phi(z)$ has no poles. The inverse Z-transform of $\psi(z) \triangleq {{\phi\left( z \right)}}/{{\prod\nolimits_{k = 1}^K {\left( {z - {z_k}} \right)} }}$ can be written in a compact form as
\begin{equation}\label{eqn:z_trans_fz}
{\cal Z}^{-1}\left\{ {{\psi}\left( z \right)} \right\}\left( t \right) = \frac{1}{{2\pi {\rm{i}}}}\oint_{\cal C} {\psi(z){z^{t - 1}}dz}  = \frac{{\det {\bf{F}}}}{{\det {\bf{Z}}}},
\end{equation}
where ${\rm i} = \sqrt{-1}$, $\mathcal Z^{-1}(\cdot)$ denotes the inverse Z-transform, ${\mathcal C}$ is the contour path of the integration which encircles all of the poles $z_1,\cdots,z_K$,
${\bf{F}} = \left[ \left({z_i}^{j-1}\right)_{i,j} \right]$ is a Vandermonde matrix of size ${K\times K}$ and ${\bf{Z}}=\left[ \left({z_i}^{j-1}\right)_{i,1\le j\le K-1}, \left({{z_i}^{t-1}}{\phi(z_i)}\right)_{i,K} \right]$ with size ${K\times K}$ and $\det \cdot$ refers to the determinant operation.
\end{theorem}
\begin{proof}
Please see Appendix \ref{app:inv_z_mat}.
\end{proof}
The second theorem is invoked to prove the monotonicity of the effective capacity with respect to the QoS exponent $\theta$.
\begin{theorem}\label{the:dec_mon}
If $f(x)$ is twice differentiable and $f(0)=0$, $\eta(x) = f(x)/x$ is increasing whenever $f^{\prime\prime} (x)\ge 0$ and decreasing otherwise.
\end{theorem}
\begin{proof}
Please see Appendix \ref{app:dec_mon}.
\end{proof}
\section{Effective Capacity of the Renewal Process with Constant Reward}\label{sec:ana}
If each time a renewal occurs with a constant reward $R$, the total amount of the accumulated reward is $S_t = R N_t$, where $N_t$ denotes the renewal counting process. More precisely, the counting process $N_t$ is defined as ${N_t} = \max \left\{ {n:{\mathcal S_n} \le t} \right\}$, where ${\mathcal S_n} = \sum\nolimits_{i = 1}^n {{X_i}} $ and $\{X_i,i\in \mathbb N\}$ is an independent and identically distributed (i.i.d.) sequence of interarrival times. Accordingly, the effective capacity of the renewal service process reduces to \cite{li2015throughput}
\begin{equation}\label{eqn:effcap_def}
{C_e} =  - \mathop {\lim }\limits_{t \to \infty } \frac{1}{{\theta t}}\ln {\mathbb{E}}\left\{ {{e^{ - \theta  R{N_t}}}} \right\}.
\end{equation}
We denote by $f_X(x)$ the probability density function (PDF) of each $X_i$. Moreover, ${\mathbb{E}}\left\{ {{e^{ - \theta R {N_t}}}} \right\}$ is the expectation taken over $N_t$, such that
\begin{align}\label{eqn:mgf_e_d}
{\mathbb{E}}\left\{ {{e^{ - \theta R {N_t}}}} \right\}&= \sum\limits_{n = 0}^\infty  {{e^{ - \theta Rn}}\Pr \left( {{N_t} = n} \right)},
\end{align}
and $\Pr \left( {{N_t} = n} \right)$ is given by
\begin{align}\label{eqn:Nt_pmf}
\Pr \left( {{N_t} = n} \right) &= \Pr \left\{ {{\mathcal S_n} \le t} \right\} - \Pr \left\{ {{\mathcal S_{n + 1}} \le t} \right\}\notag\\
 &= {F_n}\left( t \right) - {F_{n + 1}}\left( t \right),
\end{align}
where ${F_n}\left( t \right)$ represents the cumulative distribution function (CDF) of $\mathcal S_n$, i.e.,
\begin{equation}\label{eqn:F_k_def}
{F_n}\left( t \right) = \Pr \left\{ {{\mathcal S_n} \le t} \right\} = \Pr \left\{ {\sum\limits_{i = 1}^n {{X_i}}  \le t} \right\}.
\end{equation}
\subsection{Laplace Transform-Based Analysis}
The independence among $X_i$'s motivates us to calculate \eqref{eqn:F_k_def} through Laplace transform. Denote by $\mathcal F(s)$ the Laplace transform of $X$, i.e., $\mathcal F(s) = \mathbb E(e^{-sX})=\int\nolimits_0^\infty e^{-sx}f_X(x)dx$. By applying the convolutional property of Laplace transform to \eqref{eqn:F_k_def}, it follows that ${F_n}\left( t \right) = {\mathcal L^{ - 1}}\left\{ {{\left({{\mathcal F}\left( s \right)}\right)^n}/{s}} \right\}(t)$,
where ${\mathcal L^{ - 1}}$ stands for the operator of the inverse Laplace transform. By substituting this result into \eqref{eqn:Nt_pmf}, we get $\Pr \left( {{N_t} = n} \right) = {\mathcal L^{ - 1}}\left\{ {{\left({{\mathcal F}\left( s \right)}\right)^n(1-\mathcal F(s))}/{s}} \right\}(t)$. 
Putting this result into \eqref{eqn:mgf_e_d} together with the formula of the sum of a geometric series, we reach
\begin{equation}\label{eqn:mgf_e_lp}
{\mathbb{E}}\left\{ {{e^{ - \theta R {N_t}}}} \right\} 
={\mathcal L^{ - 1}}\left\{ {\frac{{1 - \mathcal F\left( s \right)}}{{s\left( {1 - {e^{ - \theta R}}\mathcal F\left( s \right)} \right)}}} \right\}\left( t \right),
\end{equation}
where $\left| {{e^{ - \theta R }}\mathcal F\left( s \right)} \right| < 1$. By substituting \eqref{eqn:mgf_e_lp} into \eqref{eqn:effcap_def}, the evaluation of the effective capacity can thus be enabled. By taking the Poisson service process as an example, $N_t$ follows a Poisson distribution with mean $\lambda t$, $\mathbb E\left\{ {{e^{ - \theta R{N_t}}}} \right\} = \sum\nolimits_{k = 0}^\infty  {{e^{ - \theta Rk}}{e^{ - \lambda t }}{(\lambda t) ^k}/k!}={\exp{\left( - \lambda t  + \lambda t{e^{ - \theta R}} \right)}}$. On the other hand, the interarrival time is found to be exponentially distributed with mean $1/(\lambda t)$. Thus, it follows that $\mathcal F\left( s \right) = \lambda t/(s+\lambda t)$. By using \eqref{eqn:mgf_e_lp}, ${\mathbb{E}}\left\{ {{e^{ - \theta R {N_t}}}} \right\} = {\mathcal L^{ - 1}}\left\{ {{{1}/{{\left(s + \lambda t  - \lambda t {e^{ - \theta R}}\right)}}}} \right\}\left( t \right)={\exp{\left( - \lambda t  + \lambda t{e^{ - \theta R}} \right)}}$, which consequently justifies the analytical result.

Unfortunately, the application of Laplace transform to HARQ systems ($F(s) $ is written in the form of $ \sum q_k e^{-sk}$) usually results in an infinite number of poles in \eqref{eqn:mgf_e_lp}, and most of the time there is no closed-form expression for the effective capacity. Moreover, the computation of \eqref{eqn:mgf_e_lp} relying on numerical methods entails higher complexity and accuracy due to the calculation of $\exp(-sk)$ for extremely high $s$. These facts hinder the thorough analysis of the effective capacity. However, if the sequence $\{X_i,i\in \mathbb N\}$ is a discrete renewal process (e.g., HARQ channel service), Z-transform can be applied to investigate the effective capacity for more insights.

\subsection{Z-Transform-Based Analysis}\label{sec:zt}
If $X_1,X_2,\cdots,X_i,\cdots$ are i.i.d. discrete and non-negative integer random variables with the probability mass function (pmf) $\Pr(X_i=k)=q_k,\,k\in[0,K]$, Z-transform of the distribution of $X_i$ is expressed as ${\mathcal X}(z) = \sum\nolimits_{k = 0}^K {{q_k}{z^{ - k}}}$. Similarly, by applying the convolutional property of Z-transform to \eqref{eqn:F_k_def}, we have ${F_n}\left( t \right) = \mathcal Z^{-1}\left\{{(\mathcal X(z))}^n{z}/{(z-1)}\right\}$. 
Combining the latter with \eqref{eqn:Nt_pmf} leads to $ \Pr \left( {{N_t} = n} \right) = {{\cal Z}^{ - 1}}\left\{ {\frac{z}{{z - 1}}{{({\cal X}(z))}^n}\left( {1 - {\cal X}(z)} \right)} \right\}$. 
From \eqref{eqn:mgf_e_d}, ${\mathbb{E}}\left\{ {{e^{ - \theta R {N_t}}}} \right\}$ can be rewritten as
\begin{align}\label{eqn:mgf_nt_rew}
{\mathbb{E}}\left\{ {{e^{ - \theta R {N_t}}}} \right\} 
  = \frac{1}{{2\pi {\rm{i}}}}\oint\nolimits_{\mathcal C} {\frac{{z\left( {1 - {\cal X}(z)} \right)}}{{\left( {z - 1} \right)\left( {1 - {e^{ - \theta R  }}{\cal X}(z)} \right)}}{z^{t - 1}}dz}.
\end{align}
where $\left| {{e^{ - \theta R }}{\cal X}(z)} \right| < 1$.
By using the definition of ${\cal X}(z)$, ${\mathbb{E}}\left\{ {{e^{ - \theta R {N_t}}}} \right\}$ can be further expressed as
\begin{align}\label{eqn:mgf_nt_inv_sim1}
{\mathbb{E}}\left\{ {{e^{ - \theta R {N_t}}}} \right\}
 =\frac{ {1 - e ^{\theta R}} }{1-q_0e^{-\theta R}}\frac{1}{{2\pi {\rm{i}}}}\oint\nolimits_{\mathcal C} {\frac{\frac{z^{t+K}\left( {1 - {\cal X}(z)} \right)}{z - 1}}{\prod\limits_{ i\in \Phi} {\left( {z - {z_i}} \right)}}dz},
\end{align}
where $\Phi \triangleq \left\{ {{z_i}:{z_i}^K \left(1- {e^{ - \theta R  }}{\cal X}(z_i) \right) = 0}\right\}$. 
Clearly, $z_i \ne 0$ if $q_K \ne 0$, $\Phi$ is therefore equivalent to $\left\{ {{z_i}:{\cal X}(z_i) = e^{\theta R}}\right\}$. Furthermore, $0<|z_i|<1$ can be proved by applying triangle inequality to ${\cal X}(z_i) = e^{\theta R} $. We assume that $z_i$'s are distinct for analytical tractability. Thus the cardinality of $\Phi$ is $K$, ${z}^K \left(1- {e^{ - \theta R  }}{\cal X}(z) \right) = (1-q_0e^{-\theta R})\prod\nolimits_{k = 1}^K {\left( {z - {z_k}} \right)} $ then follows. The integrand in \eqref{eqn:mgf_nt_inv_sim1} has $K$ distinct poles $z_i \in \Phi$, and it is worthwhile to mention that $z=1$ is not a pole because $\lim\nolimits_{z\to 1}(1-{\cal X}(z))/(z-1) = {\rm constant}$. Accordingly, by using Theorem \ref{the:inv_z_mat},
${\mathbb{E}}\left\{ {{e^{ - \theta R {N_t}}}} \right\}$ is obtained as
\begin{align}\label{eqn:mgf_rew_simpres_re1}
 {\mathbb{E}}\left\{ {{e^{ - \theta R {N_t}}}} \right\} 
 = \frac{ {1 - e ^{\theta R}} }{1-q_0e^{-\theta R}}\frac{{\det {\bf{B}} }}{{\det{\bf{A}} }},
\end{align}
where ${\bf{A}} = \left[ \left({z_i}^{j-1}\right)_{i,j} \right]$ is a Vandermonde matrix of size ${K\times K}$ and ${\bf{B}}=\left[ \left({z_i}^{j-1}\right)_{i,1\le j\le K-1}, \left({{z_i}^{t+K}}/{(z_i-1)}\right)_{i,K} \right]$ with size ${K\times K}$. 
Plugging \eqref{eqn:mgf_rew_simpres_re1} into \eqref{eqn:effcap_def} leads to
\begin{align}\label{eqn:effcap_rew}
{C_e} &=
 - \mathop {\lim }\limits_{t \to \infty } \frac{\ln \frac{\left( {1 - {e^{\theta R} }} \right){\det {\bf{B}} }}{\left({1-q_0e^{-\theta R}}\right){\det {\bf{A}} }}}{{\theta t}}=  - \mathop {\lim }\limits_{t \to \infty } \frac{1}{{\theta t}}\ln \det {\bf{B}} . 
\end{align}
Without loss of generality, let $z_1$ be the largest in absolute value in $\Phi$, i.e., $|z_1|=\max\{{|z_i|},i\in [1,K]\}$, we have
\begin{align}\label{eqn:c_e_fin}
{C_e} &=  - \mathop {\lim }\limits_{t \to \infty } \frac{\ln {z_1}^{t + K}}{{\theta t}}\left| {\left({z_i}^{j-1}\right)_{i,1\le j\le K-1}, \left(\frac{{\left(\frac{z_i}{z_1}\right)}^{t+K}}{z_i-1}\right)_{i,K}} \right|\notag\\
 &=  \frac{{\ln {z_1}^{-1}}}{\theta }=\frac{\ln\zeta}{\theta},
\end{align}
where $\zeta \triangleq {z_1}^{-1}$. In the following, some discussions are expanded on the basis of \eqref{eqn:c_e_fin}.
\subsubsection{Calculation of $\zeta$}\label{sec:cal}
It is clear from \eqref{eqn:c_e_fin} that $\zeta \ge 1$, and $\zeta$ satisfies
\begin{equation}\label{eqn:z_1_root_eq}
\sum\limits_{k = 0}^K {{q_k}{{{\zeta}}^{k}}}  = {e^{\theta R}}.
\end{equation}
By applying Jensen's inequality to \eqref{eqn:z_1_root_eq}, we have
\begin{equation}\label{eqn:jesen_ineq}
\sum\limits_{k = 0}^K {{q_k}{{\zeta}^{ k}}}  \ge {\zeta^{ \sum\limits_{k = 0}^K {k{q_k}} }} = {{\zeta}^{ \mathbb E\left( X \right)}}.
\end{equation}
Thus, $\zeta \le {e^{\frac{\theta R }{{\mathbb E\left( X \right)}}}}$ follows. Note $\zeta$ is larger than or equal to 1, $\zeta$ belongs to $ [1,{e^{\frac{\theta R }{{\mathbb E\left( X \right)}}}}]$. The latter eases the calculation of $\zeta$. More specifically, let us define $g(x )\triangleq\sum\nolimits_{k = 0}^K {{q_k}{{x }^{ k}}}-{e^{\theta R}}$. Since $g(x )$ is an increasing function of $x $ if $x >0$, $\zeta$ is definitely the unique zero point of $g(x )$ within the range $x  \in [1,{e^{\frac{\theta R }{{\mathbb E\left( X \right)}}}}]$. Accordingly, the bisection method can be adopted to calculate $\zeta$.

\subsubsection{Approximation of ${C_e}$}\label{sec:app_ce_con}
If $\theta$ is sufficiently small, the effective capacity can be approximated by $C_e\approx{R}/{{\mathbb E\left( X \right)}}-{\theta R^2{{\rm Cov}\left( X \right)}}/\left({2{\mathbb E\left( X \right)}^3}\right)$ from \cite{li2015throughput}.
This can be further confirmed by using Taylor expansion of \eqref{eqn:c_e_fin}. To proceed, we define $u(\theta) = \ln \zeta$ and $h\left( u \right) = \sum\nolimits_{k = 0}^K {{q_k}{e^{ku}}} $ such that
\begin{equation}\label{eqn:u_h_re}
h\left( u(\theta) \right) = {e^{\theta R}}.
\end{equation}
By using Taylor series expansion, $u(\theta)$ can be expanded as
\begin{equation}\label{eqn:u_expan}
u(\theta ) = u\left( 0 \right) + u'\left( 0 \right)\theta  + \frac{{u''\left( 0 \right)}}{{2}}{\theta ^2}  + o\left( {{\theta ^2}} \right),
\end{equation}
where $o(\cdot)$ refers to the little-O notation, $u'$ and $u''$ denote the first and second derivatives of $u$ with respect to (w.r.t.) $\theta$, respectively. $u\left( 0 \right)$ and its derivatives can be determined as follows.

It is readily found from \eqref{eqn:u_h_re} that ${u\left( 0 \right)}=1$. Taking the first derivative w.r.t. $\theta$ at the both sides of \eqref{eqn:u_h_re} leads to $\frac{{dh}}{{du}}u'(\theta ) = R{e^{\theta R}}$.
With the definition of $h(u)$, $u'(\theta )$ can then be obtained as $u'(\theta )= {{R{e^{\theta R}}}}\left/
 {\vphantom {{} {}}}\right.{{\sum\nolimits_{k = 0}^K {{q_k}k{e^{ku\left( \theta  \right)}}} }}$.
Hence, $u'(0) = {R}/{{\mathbb E\left( X \right)}}$. Similarly, taking the first derivative of $u'(\theta )$ w.r.t. $\theta$ gives $u''(\theta )$ as
\begin{multline}\label{eqn:u_sec_der}
u''(\theta ) = {R^2}{e^{\theta R}}{{{{\left( {\sum\limits_{k = 0}^K {{q_k}k{e^{ku\left( \theta  \right)}}} } \right)}^{-3}}}}\times \\
\left({{{{\left( {\sum\limits_{k = 0}^K {{q_k}k{e^{ku\left( \theta  \right)}}} } \right)}^2} - \sum\limits_{k = 0}^K {{q_k}{e^{ku\left( \theta  \right)}}} \sum\limits_{k = 0}^K {{q_k}{k^2}{e^{ku\left( \theta  \right)}}} }}\right).
\end{multline}
Thus, $u''\left( 0 \right) = -{{R^2}{{\rm{Cov}}\left( X \right)}}/{{{{\left( {\mathbb E\left( X \right)} \right)}^3}}}$, where ${\rm{Cov}}(\cdot)$ denotes the covariance operator.
Substituting ${u\left( 0 \right)}$, $u'(0)$ and $u''(0)$ into \eqref{eqn:u_expan} together with the definition of $u(\theta)$ arrives at $\zeta \approx {\exp\left({{\theta R }/{{\mathbb E\left( X \right)}}-{\theta^2 R^2{{\rm Cov}\left( X \right)}}/{\left(2{\mathbb E\left( X \right)}^3\right)}}\right)}$.
Consequently, putting $\zeta$ into \eqref{eqn:c_e_fin} leads to the approximation of $C_e$, which coincides with the result in \cite{li2015throughput}.
\subsubsection{Properties of Effective Capacity}\label{sec:prop_ec}
By applying the Cauchy-Schwarz inequality to \eqref{eqn:u_sec_der}, we have $u^{\prime \prime}(\theta)\le 0$.
With Theorem \ref{the:dec_mon}, the effective capacity is found to be a decreasing function of QoS exponent $\theta$. Moreover, as proved in Appendix \ref{sec:app_prop}, ${C_e}$ is bounded as
\begin{equation}\label{eqn:ce_bounds_con}
\frac{R}{K} \le {C_e} \le \min\left(\frac{R}{K} - \frac{{\ln {q_K}}}{{K\theta }},\frac{R}{{\mathbb E}(X)}\right),
\end{equation}
and $\mathop {\lim }\nolimits_{\theta  \to \infty} C_e={R}/{K}$ if $q_K > 0$.

\subsubsection{Extension to the General Case}\label{sec:ext_gen}
The result of \eqref{eqn:c_e_fin} can be extended to a general case when $\{X_i\}$ is a non-negative continuous renewal process. Specifically, the horizontal axis of the distribution of $X_i$ can be partitioned into a number of equal intervals, each with length $\Delta x$. This discretization leads to a new discrete random variable $\tilde X_i$ with pmf given by $\Pr(\tilde X_i=k)=\int\nolimits_{k\Delta x}^{(k+1)\Delta x}f_X(x)dx \triangleq q_k$. Therefore, a similar approach in Subsection \ref{sec:zt} can be adopted to approximate the effective capacity, and the approximation can become more accurate by using a higher-resolution discretization. As $\Delta x \to 0$, it is readily proved that the exact expression of the effective capacity is
\begin{equation}\label{eqn:ec_cons_renw_var}
  C_e = \frac{\ln\zeta}{\theta},
\end{equation}
where $\zeta$ is the solution to the following equation
\begin{equation}\label{eqn:con_ec_ext}
\mathbb E\left( {{{\zeta }^{ X}}} \right) = {e^{\theta R}},\, \zeta \ge 0.
\end{equation}
$\zeta  \in [1,{e^{\frac{\theta R }{{\mathbb E\left( X \right)}}}}]$, and the proof can be found in Appendix \ref{app:proof_cn_ec_ext}. Furthermore, it is readily proved that the same properties as shown in Subsections \ref{sec:cal}-\ref{sec:prop_ec} also apply to the general case except for the bounds of the effective capacity in \eqref{eqn:ce_bounds_con}. Nevertheless, \eqref{eqn:ce_bounds_con} is applicable to the discrete renewal service process with non-integer interarrival time, and $K$ stands for the maximum interarrival time herein.

\section{Effective Capacity of the General Renewal Reward Process}\label{sec:ec_gen}
In general, the reward earned each time could be variable apart from constant \cite{ross2014introduction}. In communication systems, the reward of the $i$-th renewal commonly varies with the length of the renewal interval $X_i$ and the channel/termination state $\frak S_i$ of each renewal\footnote{For example, $\frak S_i$ represents the outcome whether the receiver succeeds to decode the message after the termination of HARQ. Moreover, $\frak S_i$ could also be the channel gains if perfect CSI is known at the transmitter, and adaptive \& modulation scheme is adopted.}. We denote by $(X_i,\frak S_i)$ the $i$-th renewal event state, and denote by $\mathcal R(X_i,\frak S_i)$ the reward received from the $i$-th renewal. Likewise, the effective capacity of discrete renewal service processes is derived first, and the analytical results are further extended to that of continuous ones.

By favor of the definition of the renewal counting process $N_t$, $S_t$ is given by
\begin{equation}\label{eqn:S_t_var}
S_t = \sum\limits_{i = 1}^{N_t} {{\mathcal R(X_i,\frak S_i)}}.
\end{equation}
Note that the renewal depends on the interarrival time and the termination state, we define $R_{k,s} \triangleq \mathcal R(X_i=k,\frak S_i=s)$ for notational convenience, where $k \in [1,K]$ and $s \in [1,v_k]$. Moreover, denote by $q_{k,s}$ the probability that the renewal event is ended with interval length $k$ and termination state $s$, i.e., $q_{k,s} \triangleq {\Pr \left( X_i=k,\frak S_i=s \right)}$. Thus, the probability that the interarrival time is $k$ can be obtained as
\begin{align}\label{eqn:inter_time_prob}
\Pr \left( X_i=k \right) &= {\sum\limits_{j = 1}^{{v_i}} {\Pr \left( {{X_i} = k,\frak S_i = j} \right)} }={\sum\limits_{j = 1}^{{v_i}} {q_{k,j}} }.
\end{align}

Adding up the probabilities $q_{k,s}$ of all the renewal states equals to one, i.e., 
\begin{equation}\label{eqn:totalProb_}
\sum\nolimits_{k = 1}^K {\sum\nolimits_{s = 1}^{{v_k}} {{q_{k,s}}} }  = 1.
\end{equation}
To facilitate the analysis, we define a vector of renewal counting processes that count the numbers of $R_{k,s}$'s achieved by the network service up to time $t$ as ${{\bf{n}}_t} = \left( {{{\left( {{n_{t,k,1}}, \cdots ,{n_{t,k,{v_k}}}} \right)}_{k=1}^K}} \right)$, where $n_{t,k,s}$ represents the number of the reward $R_{k,s}$'s earned up until time $t$, and the subscript $t$ is omitted in the sequel for simplicity. Moreover, since the maximum interval between two consecutive renewals is $K$, ${{\bf{n}}_t}$ should satisfy the following constraint as
\begin{equation}\label{eqn:nt_cons}
{\left[ {t - K + 1} \right]^ + } \le \sum\nolimits_{k = 1}^K {\sum\nolimits_{s = 1}^{{v_k}} {k{n_{k,s}}} }  \le t,
\end{equation}
where $[x]^+=\max\{0,x\}$ is the projection onto the nonnegative orthant. With the above definition, the total amount of reward given in \eqref{eqn:S_t_var} can be rewritten as
\begin{equation}\label{eqn:S_t_var_rew}
{S_t} = \sum\nolimits_{k = 1}^K {\sum\nolimits_{s = 1}^{{v_k}} {{n_{k,s}}{R_{k,s}}} }.
\end{equation}
Therefore, the corresponding effective capacity can be obtained as
\begin{equation}\label{eqn:effcap_def_var}
{C_e} =  - \mathop {\lim }\limits_{t \to \infty } \frac{1}{{\theta t}}\ln \underbrace {\mathbb E\left\{ {{e^{ - \theta \sum\nolimits_{k = 1}^K {\sum\nolimits_{s = 1}^{{v_k}} {{n_{k,s}}{R_{k,s}}} } }}} \right\}}_{\varphi \left( t \right)},
\end{equation}
where ${\varphi \left( t \right)}$ is given by \eqref{eqn:avg_ec_vari}, shown at the top of the next page.
\begin{figure*}[!t]
\begin{equation}\label{eqn:avg_ec_vari}
{\varphi \left( t \right)}=\sum\limits_{\sum\limits_{k = 1}^K {\sum\limits_{s = 1}^{{v_k}} {k{n_{k,s}}} }  \in \left[ {{{\left[ {t - K + 1} \right]}^ + },t} \right]} {{e^{ - \theta \sum\limits_{k = 1}^K {\sum\limits_{s = 1}^{{v_k}} {{n_{k,s}}{R_{k,s}}} } }}\Pr \left( {{\bf{n}}_t = \left( {{{\left( {{n_{k,1}}, \cdots ,{n_{k,{v_k}}}} \right)}_{k=1}^K}} \right)} \right)}.
\end{equation}
\end{figure*}
Different from Section \ref{sec:ana}, it is impossible to directly apply Z-transform to derive the moment generating function (mgf) of $S_t$, $\varphi(t)$, due to the complex vector form of ${{\bf{n}}_t}$. To address this challenge, the effective capacity with finite time is obtained by using the multinomial distribution of ${{\bf{n}}_t}$.

\subsection{Effective Capacity with Finite Time}
To proceed, $\Pr \left( {{{\bf{n}}_t} = \left( {{{\left( {{n_{k,1}}, \cdots ,{n_{k,{v_k}}}} \right)}_{k=1}^K}} \right)} \right)$ is derived first. To do so, we define $\tau = t - \sum\nolimits_{k = 1}^K {\sum\nolimits_{s = 1}^{{v_k}} {k{n_{k,s}}} } $ to show the time of the last renewal passed before time $t$. Clearly, it follows from \eqref{eqn:nt_cons} that $\tau$ is bounded as $0 \le \tau \le K$, and the interval of the last renewal should be larger than $\tau$, i.e., $X_{N_t+1} > \tau$. 
 According to the law of total probability, $\Pr \left( {{{\bf{n}}_t} = \left( {{{\left( {{n_{k,1}}, \cdots ,{n_{k,{v_k}}}} \right)}_{k=1}^K}} \right)} \right)$ is given by \eqref{eqn:n_t_jointpdf_fin}, shown at the top of the page, where step (a) holds because of the independence among renewals, and step (b) holds by using \eqref{eqn:inter_time_prob}.
\begin{figure*}[!t]
\begin{align}\label{eqn:n_t_jointpdf_fin}
\Pr \left( {{{\bf{n}}_t} = \left( {{{\left( {{n_{k,1}}, \cdots ,{n_{k,{v_k}}}} \right)}_{k=1}^K}} \right)} \right) &= \sum\limits_{i  = \tau +1}^K {\Pr \left( {\left. {{{\bf{n}}_t} = \left( {{{\left( {{n_{k,1}}, \cdots ,{n_{k,{v_k}}}} \right)}_{k=1}^K}} \right)} \right|{X_{{N_t} + 1}} = i } \right)\Pr \left( {{X_{{N_t} + 1}} = i } \right)} \notag\\
 &\mathop  = \limits^{\left( a \right)}  \Pr \left( {{{\bf{n}}_{t - \tau }} = \left( {{{\left( {{n_{k,1}}, \cdots ,{n_{k,{v_k}}}} \right)}_{k=1}^K}} \right)} \right)\sum\limits_{i  = \tau+1 }^K {\Pr \left( {X_{{N_t} + 1}  = i } \right)}\notag\\
 &\mathop  = \limits^{\left( b \right)} \Pr \left( {{{\bf{n}}_{t - \tau }} = \left( {{{\left( {{n_{k,1}}, \cdots ,{n_{k,{v_k}}}} \right)}_{k=1}^K}} \right)} \right)\sum\limits_{i = \tau+1 }^K {\sum\limits_{j = 1}^{{v_i}} {{q_{i,j}}} } ,
\end{align}
\end{figure*}
More specifically, the vector ${\bf n}_{t-\tau}$ follows a multinomial distribution with the pmf given by \eqref{eqn:n_t_cond}, shown at the top of this page, where ${{n}\choose {{n_1}, \cdots ,{n_Q}}}= {{n!}}/\left({{{n_1}! \cdots {n_Q}!}}\right)$ is the multinomial coefficient and $n=\sum\nolimits_{q = 1}^Q {{n_q}}$.
\begin{figure*}[!t]
\begin{align}\label{eqn:n_t_cond}
\Pr \left( {{{\bf{n}}_{t - \tau }} = \left( {{{\left( {{n_{k,1}}, \cdots ,{n_{k,{v_k}}}} \right)}_{k=1}^K}} \right)} \right) = \left( {\begin{array}{*{20}{c}}
{\sum\limits_{k = 1}^K {\sum\limits_{s = 1}^{{v_k}} {{n_{k,s}}} } }\\
{ {{{\left( {{n_{k,1}}, \cdots ,{n_{k,{v_k}}}} \right)}_{k=1}^K}}}
\end{array}} \right)\prod\limits_{k = 1}^K {\prod\limits_{s = 1}^{{v_k}} {{q_{k,s}}^{{n_{k,s}}}} }.
\end{align}
\hrulefill
\end{figure*}

By substituting \eqref{eqn:n_t_jointpdf_fin} and \eqref{eqn:n_t_cond} into \eqref{eqn:avg_ec_vari} together with some algebraic manipulations, $\varphi\left( t \right)$ is finally rewritten as \eqref{eqn:phi_fin} at the top of the following page.
\begin{figure*}[!t]
\begin{align}\label{eqn:phi_fin}
&\varphi\left( t \right) = \notag\\
&\sum\limits_{\sum\limits_{k = 1}^K {\sum\limits_{s = 1}^{{v_k}} {k{n_{k,s}}} }  \in \left[ {{{\left[ {t - K + 1} \right]}^ + },t} \right]} {{e^{ - \theta \sum\limits_{k = 1}^K {\sum\limits_{s = 1}^{{v_k}} {{R_{k,s}}{n_{k,s}}} } }}\left( {\begin{array}{*{20}{c}}
{\sum\limits_{k = 1}^K {\sum\limits_{s = 1}^{{v_k}} {{n_{k,s}}} } }\\
{\left( {{n_{k,1}}, \cdots ,{n_{k,{v_k}}}} \right)_{k = 1}^K}
\end{array}} \right)\prod\limits_{k = 1}^K {\prod\limits_{s = 1}^{{v_k}} {{q_{k,s}}^{{n_{k,s}}}} } \sum\limits_{i = t - \sum\limits_{k = 1}^K {\sum\limits_{s = 1}^{{v_k}} {k{n_{k,s}}} }  + 1}^K {\sum\limits_{j = 1}^{{v_i}} {{q_{i,j}}} } }.
\end{align}
\end{figure*}
Specifically, if $0 \le t<K$, \eqref{eqn:phi_fin} becomes
\begin{align}\label{eqn:varphi_t_lessK}
\varphi \left( t \right) &= \sum\limits_{\sum\limits_{k = 1}^t {\sum\limits_{s = 1}^{{v_k}} {k{n_{k,s}}} }  \in \left[ {0,t} \right]} {{e^{ - \theta \sum\limits_{k = 1}^t {\sum\limits_{s = 1}^{{v_k}} {{R_{k,s}}{n_{k,s}}} } }}} \notag\\
&\quad\times \left( {\begin{array}{*{20}{c}}
{\sum\limits_{k = 1}^t {\sum\limits_{s = 1}^{{v_k}} {{n_{k,s}}} } }\\
{\left( {{n_{k,1}}, \cdots ,{n_{k,{v_k}}}} \right)_{k = 1}^t}
\end{array}} \right)\prod\limits_{k = 1}^t {\prod\limits_{s = 1}^{{v_k}} {{q_{k,s}}^{{n_{k,s}}}} } \notag\\
&\quad \times \sum\limits_{i = t - \sum\limits_{k = 1}^t {\sum\limits_{s = 1}^{{v_k}} {k{n_{k,s}}} }  + 1}^K {\sum\limits_{j = 1}^{{v_i}} {{q_{i,j}}} },\,0 \le t< K,
\end{align}
and $\varphi(0)=1$. However, for large $t$, the computation of effective capacity with \eqref{eqn:phi_fin} causes higher complexity overhead. To address this issue, we turn to the Z-transform to derive the effective capacity in a closed-form.
\subsection{Z-Transform-Based Analysis}\label{sec:z_trans_var_ren}
By virtue of the recurrence relation of multinomial coefficient \cite[eq.26.4.10]{olver2010nist}, \eqref{eqn:phi_fin} can be rewritten as (\ref{eqn:phi_t_recur}), shown at the top of the next page.
\begin{figure*}[!t]
\begin{align}\label{eqn:phi_t_recur}
\varphi \left( t \right) &=   { {\sum\limits_{\sum\limits_{k = 1}^K {\sum\limits_{s = 1}^{{v_k}} {k{n_{k,s}}} }  \in \left[ {{{\left[ {t - K + 1} \right]}^ + },t} \right]} {{e^{ - \theta \sum\limits_{k = 1}^K {\sum\limits_{s = 1}^{{v_k}} {{R_{k,s}}{n_{k,s}}} } }}\prod\limits_{k = 1}^K {\prod\limits_{s = 1}^{{v_k}} {{q_{k,s}}^{{n_{k,s}}}} } \sum\limits_{i = t - \sum\limits_{k = 1}^K {\sum\limits_{s = 1}^{{v_k}} {k{n_{k,s}}} }  + 1}^K {\sum\limits_{j = 1}^{{v_i}} {{q_{i,j}}} } } } } \notag\\
&\quad \times \sum\limits_{\kappa = 1}^K\sum\limits_{\nu = 1}^{{v_\kappa}} \left( {\begin{array}{*{20}{c}}
{\sum\limits_{k = 1}^K {\sum\limits_{s = 1}^{{v_k}} {{n_{k,s}}} }  - 1}\\
{\left( {{n_{1,1}}, \cdots ,{n_{1,{v_1}}}} \right), \cdots ,\left( {{n_{\kappa,1}}, \cdots ,{n_{\kappa,\nu}} - 1, \cdots ,{n_{\kappa,{v_\kappa}}}} \right), \cdots ,\left( {{n_{K,1}}, \cdots ,{n_{K,{v_K}}}} \right)}
\end{array}} \right).
\end{align}
\hrulefill
\end{figure*}
By assuming $t\ge K$ and switching the order of summations along with some rearrangements, $\varphi \left( t \right)$ is represented by a homogeneous difference equation as
\begin{equation}\label{eqn:varphi_heq}
\varphi\left( t \right) - \sum\limits_{\kappa = 1}^K {{a_\kappa}\varphi\left( {t - \kappa} \right)} = 0,\,t\ge K,
\end{equation}
where $a_\kappa = {\sum\nolimits_{j = 1}^{{v_\kappa}} {{q_{\kappa,j}}{e^{ - \theta {R_{\kappa,j}}}}} }$. 
As proved in Appendix \ref{app:proof_varphi_zf}, by applying Z-transform to \eqref{eqn:varphi_heq}, $\varphi\left( t \right)$ can be expressed as
\begin{equation}\label{eqn:varphi_zt}
\varphi(t) = \sum\limits_{l  = 1}^K {  {\varphi\left( {K - l} \right)\frac{{\det {\tilde{\bf{B}}_{l}}}}{{\det \tilde{\bf{A}}}}} },\,t\ge K,
\end{equation}
where $\tilde{\bf{A}} = \left[ {{{\left( {{\tilde z_i}^{j - 1}} \right)}_{i,j}}} \right]$, ${\tilde{\bf{B}}_{l}} = \left[ {{{\left( {{\tilde z_i}^{j - 1}} \right)}_{i,1 \le j \le K - 1}},{{\left( {\sum\nolimits_{\kappa = l}^K {{a_\kappa}{{\tilde z_i}^{t + l - \kappa - 1}}} } \right)}_{i,K}}} \right]$, and $\tilde z_i$'s are $K$ distinct roots of the polynomial ${ z^K} - \sum\nolimits_{\kappa = 1}^K {{a_\kappa}{ z^{K - \kappa}}}  = 0$. Without loss of generality, we define $|\tilde z_1|=\max\{{|\tilde z_i|},i\in [1,K]\}$. Plugging \eqref{eqn:varphi_zt} into \eqref{eqn:effcap_def_var} leads to
\begin{align}\label{eqn:ce_var_limit}
&{C_e} =  - \mathop {\lim }\limits_{t \to \infty } \frac{1}{{\theta t}}\ln {{\tilde z}_1}^t\sum\limits_{l = 1}^K {\varphi \left( {K - l} \right){{\tilde z}_1}^{l - \kappa  - 1}} \notag\\
&\quad \times \frac{{\det \left[ {{{\left( {{{\tilde z}_i}^{j - 1}} \right)}_{i,1 \le j \le K - 1}},{{\left( {\sum\limits_{\kappa  = l}^K {{a_\kappa }{{\left( {\frac{{{{\tilde z}_i}}}{{{{\tilde z}_1}}}} \right)}^{t + l - \kappa  - 1}}} } \right)}_{i,K}}} \right]}}{{\det \tilde {\bf{A}}}}\notag\\
&\quad =\frac{\ln{{\tilde z}_1^{-1}}}{\theta}=\frac{\ln \tilde \zeta}{\theta},
\end{align}
where $\tilde \zeta = {\tilde z}_1^{-1}$. In analogous to Subsection \ref{sec:zt}, some discussions are carried out as follows based on \eqref{eqn:ce_var_limit}.
\subsubsection{Calculation of $\zeta$}\label{sec:cal_hat}
It is obvious from \eqref{eqn:ce_var_limit} that $\tilde \zeta$ is a real number and $\tilde \zeta \ge 1$. Similarly, we define $\tilde g(x ) = \sum\nolimits_{\kappa = 1}^K {{a_\kappa}{x^{\kappa}}} - 1$. Hence, $\tilde \zeta$ is the zero point of $\tilde g(x )$ and uniquely exists owing to its increasing monotonicity within $x \ge 0$. The monotonicity enables us to find the zero point $\tilde \zeta$ with bisection method, and $\tilde \zeta$ lies in $[1,\exp{\left({\theta{\mathbb E\left( {{{\mathcal R(X,\frak S)}}} \right)}}/{{\mathbb E\left( X \right)}}\right)}]$, where the upper bound is proved in Subsection \ref{sec:prop_ec_var}.

\subsubsection{Approximation of $C_e$}
By defining $\tilde u(\theta) = \ln \tilde \zeta$ and $\tilde h\left( { \tilde u} \right) = \sum\nolimits_{\kappa  = 1}^K {\sum\nolimits_{j = 1}^{{v_\kappa }} {{q_{\kappa ,j}}{e^{ - \theta {R_{\kappa ,j}}}}} {e^{\kappa \tilde u}}}$, we have $\tilde h\left( {\tilde u} \right)=1$. In analogous to Subsection \ref{sec:app_ce_con}, $\tilde u(\theta)$ can be approximated for small $ \theta$ as
\begin{equation}\label{eqn:tilde_u_expan}
\tilde u(\theta ) \approx \tilde u\left( 0 \right) + \tilde u'\left( 0 \right)\theta  + \frac{{\tilde u''\left( 0 \right)}}{{2}}{\theta ^2},
\end{equation}
where $\tilde u'$ and $\tilde u''$ denote the first and second derivatives of $\tilde u$ w.r.t. $\theta$, respectively. Clearly, $\tilde u(0)=0$. $\tilde u'(\theta)$ and $\tilde u''(\theta)$ are given by
\begin{equation}\label{eqn:u_tilde_fd}
\tilde u'(\theta) = \frac{{\sum\nolimits_{\kappa = 1}^K {\sum\nolimits_{j = 1}^{{v_\kappa}} {{q_{\kappa,j}}{R_{\kappa,j}}{e^{ - {R_{\kappa,j}}\theta }}{e^{\kappa \tilde u}}} } }}{{\sum\nolimits_{\kappa = 1}^K {\sum\nolimits_{j = 1}^{{v_\kappa}} {\kappa{q_{\kappa,j}}{e^{ - {R_{\kappa,j}}\theta }}{e^{\kappa \tilde u}}} } }},
\end{equation}
and \eqref{eqn:u_tilde_sec} at the top of the next page, respectively.
\begin{figure*}[!t]
\begin{equation}\label{eqn:u_tilde_sec}
\tilde u''(\theta) = \frac{{\left( \begin{array}{l}
\left( { - \sum\limits_{\kappa = 1}^K {\sum\limits_{j = 1}^{{v_\kappa}} {{q_{\kappa,j}}{R_{\kappa,j}}^2{e^{ - {R_{\kappa,j}}\theta }}{e^{\kappa \tilde u}}} }  + \tilde u'\sum\limits_{\kappa = 1}^K {\sum\limits_{j = 1}^{{v_\kappa}} {\kappa{q_{\kappa,j}}{R_{\kappa,j}}{e^{ - {R_{\kappa,j}}\theta }}{e^{\kappa \tilde u}}} } } \right)\sum\limits_{\kappa = 1}^K {\sum\limits_{j = 1}^{{v_\kappa}} {\kappa{q_{\kappa,j}}{e^{ - {R_{\kappa,j}}\theta }}{e^{\kappa \tilde u}}} } \\
 + \left( {\sum\limits_{\kappa = 1}^K {\sum\limits_{j = 1}^{{v_\kappa}} {\kappa{q_{\kappa,j}}{R_{\kappa,j}}{e^{ - {R_{\kappa,j}}\theta }}{e^{\kappa \tilde u}}} }  - \tilde u'\sum\limits_{\kappa = 1}^K {\sum\limits_{j = 1}^{{v_\kappa}} {{\kappa^2}{q_{\kappa,j}}{e^{ - {R_{\kappa,j}}\theta }}{e^{\kappa \tilde u}}} } } \right)\sum\limits_{\kappa = 1}^K {\sum\limits_{j = 1}^{{v_\kappa}} {{q_{\kappa,j}}{R_{\kappa,j}}{e^{ - {R_{\kappa,j}}\theta }}{e^{\kappa \tilde u}}} }
\end{array} \right)}}{{{{\left( {\sum\limits_{\kappa = 1}^K {\sum\limits_{j = 1}^{{v_\kappa}} {\kappa{q_{\kappa,j}}{e^{ - {R_{\kappa,j}}\theta }}{e^{\kappa \tilde u}}} } } \right)}^2}}},
\end{equation}
\hrulefill
\end{figure*}
Hence, $\tilde u'(0)$ and $\tilde u''(0)$ can be respectively obtained as
\begin{equation}\label{eqn:u_tilde_fd0}
\tilde u'(0) = \frac{{\sum\nolimits_{\kappa  = 1}^K {\sum\nolimits_{j = 1}^{{v_\kappa }} {{q_{\kappa ,j}}{R_{\kappa ,j}}} } }}{{\sum\nolimits_{\kappa  = 1}^K {\kappa\sum\nolimits_{j = 1}^{{v_\kappa }} { {q_{\kappa ,j}}} } }} = \frac{{\mathbb E\left( {{{\mathcal R(X,\frak S)}}} \right)}}{{\mathbb E\left( X \right)}},
\end{equation}
\begin{align}\label{eqn:u_tilde_sec0}
\tilde u''(0)
&=- \frac{{\mathbb E{{\left( {{{\mathcal R(X,\frak S)}}\mathbb E\left( X \right) - \mathbb E\left( {{{\mathcal R(X,\frak S)}}} \right)X} \right)}^2}}}{{{{\left( {\mathbb E\left( X \right)} \right)}^3}}}.
\end{align}
\eqref{eqn:u_tilde_fd0} implies that $\tilde u'(0)$ refers to the average reward over time. Particularly for the HARQ schemes, $\tilde u'(0)$ is the so-called the long term average throughput (LTAT)\cite{caire2001throughput}.

Substituting \eqref{eqn:u_tilde_fd0} and \eqref{eqn:u_tilde_sec0} into \eqref{eqn:tilde_u_expan} along with \eqref{eqn:ce_var_limit}, the effective capacity can be approximated as
\begin{align}\label{eqn:ce_app_var}
{C_e} &\approx \frac{{\mathbb E\left( {{\cal R}(X,\frak S)} \right)}}{{\mathbb E\left( X \right)}} \notag\\
&\quad- \frac{{\theta \mathbb E\left\{{{\left( {{\cal R}(X,\frak S)\mathbb E\left( X \right) - \mathbb E\left( {{\cal R}(X,\frak S)} \right)X} \right)}^2}\right\}}}{2{{{\left( {\mathbb E\left( X \right)} \right)}^3}}}.
\end{align}

\subsubsection{Properties of Effective Capacity}\label{sec:prop_ec_var}
It is proved in Appendix \ref{app:pf_prop_ec} that the effective capacity is a decreasing function of $\theta$. Moreover, ${C_e}$ is bounded as
\begin{equation}\label{eqn:ec_bound_fin_var}
\frac{{{R_{\hat \kappa ,\hat j}}}}{K} \le {C_e} \le \min \left\{ {\frac{{{R_{\hat \kappa ,\hat j}}}}{{\hat \kappa }} - \frac{{\ln {q_{\hat \kappa ,\hat j}}}}{{\hat \kappa \theta }},\frac{{\mathbb E\left( {{{\mathcal R(X,\frak S)}}} \right)}}{{\mathbb E\left( X \right)}}} \right\},
\end{equation}
where $(\hat \kappa,\hat j) = \mathop {\arg }\nolimits_{\left( {\kappa ,j} \right)} \min\nolimits_{q_{\kappa ,j}>0} {R_{\kappa ,j}}$. Accordingly, as $\theta$ approaches to infinity, ${C_e}$ is bounded as
\begin{equation}\label{eqn:ec_lower_bounds_lbd}
 \frac{{{R_{\hat \kappa ,\hat j}}}}{K} \le \mathop {\lim }\limits_{\theta  \to  \infty}{C_e} \le  \frac{{{R_{\hat \kappa ,\hat j}}}}{{\hat \kappa }}.
\end{equation}
\subsubsection{Extension to the Continuous Ones}
Similarly to Subsection \ref{sec:ext_gen}, the result of \eqref{eqn:ce_var_limit} can also be extended to the case that $\{(X_i,\frak S_i)\}$ is a non-negative continuous renewal process and the reward is variable. Specifically, the horizontal and vertical axis of the distribution of $(X_i,\frak S_i)$ can be partitioned into a number of equal rectangles, each with area $\Delta x \times \Delta s$. This discretization leads to two new discrete random variables $(\tilde X_i,\tilde {\frak S}_i)$ with pmf given by $\Pr(\tilde X_i=k,\tilde {\frak S}_i=l)=\int\nolimits_{(k-1)\Delta x}^{k\Delta x}\int\nolimits_{(l-1)\Delta s}^{l\Delta s}f_{X,\frak S}(x,s)dxds \triangleq q_{k,l}$. Therefore, the similar approach in Subsection \ref{sec:z_trans_var_ren} can be employed to approximate the effective capacity. As $\Delta x , \Delta s \to 0$, we get the exact expression of the effective capacity that is the same as \eqref{eqn:ec_cons_renw_var}, where $\zeta$ is the solution to the following equation
\begin{equation}\label{eqn:var_ren_conti}
\mathbb E_{X,\frak S}\left( {{e^{ - \theta \mathcal R\left( {X,\frak S} \right)}}{\zeta ^X}} \right) = 1,\,\zeta \ge 1.
\end{equation}
Similarly, the effective capacity of the continuous case also follows same properties as shown in Subsections \ref{sec:cal_hat}-\ref{sec:prop_ec_var} except for the bounds of the effective capacity in \eqref{eqn:ec_bound_fin_var} and \eqref{eqn:ec_lower_bounds_lbd}. Nevertheless, \eqref{eqn:ec_bound_fin_var} and \eqref{eqn:ec_lower_bounds_lbd} are applicable to the discrete renewal service process with non-integer interarrival time, where $K$ is the maximum interarrival time.

\section{Applications to HARQ Systems}\label{sec:appli}

\subsection{Fixed-Rate HARQ Systems}
Since HARQ transmissions can be modelled by the renewal-reward process \cite{caire2001throughput}, the proceeding analysis of the effective capacity can be used to characterize the cross-layer throughput of HARQ\cite{li2015throughput}. To this end, a renewal of HARQ transmissions is defined as an event that the receiver successfully receives the message or the maximum number of transmissions is reached. It is assumed that the number of transmissions for each message is allowed up to ${\mathcal K}$ due to congestion avoidance. In addition, $N_t$ is the number of renewals that occur up until time $t$, and $X_k$ is the random time between two consecutive renewals.

We assume that each delivered packet contains $b$ information bits, and each one is first encoded into a long codeword. Following the FR-HARQ scheme, the generated codeword is then partitioned into $\mathcal K$ subcodewords, and each subcodeword consists of $L$ symbols. Suppose that each symbol duration is normalized to unit. More specifically, $X_i$ is the time interval demanded in the successful delivery of the $i$-th message, i.e., $X_k = k L$ and $k \le {\mathcal K}$. The pmf of $X_k$, i.e., $q_k$ is given by \cite{caire2001throughput}
\begin{align}\label{eqn:q_k_harq}
\Pr(X = k L) &\triangleq {q_k} \notag\\
&= {p_{k-1}} - {p_{k}}{\mathds 1}(k\ne {\mathcal K}),\,k \in [1,{\mathcal K}],
\end{align}
where ${p_{k}}$ is the outage probability after $k$ HARQ rounds and the indicator function ${\mathds 1}(\mathpzc A)$ is one whenever $\mathpzc A$ is true and zero otherwise. By applying capacity-achieving codes, the outage expressions for three different HARQ schemes are given respectively by \cite{caire2001throughput}
\begin{multline}
\label{eqn:outage_definition}
{p_k} = \\
\left\{ {\begin{array}{*{20}{l}}
{\Pr \left( {{{\log }_2}\left( {1 + \max \left( {{\gamma _1}, \cdots ,{\gamma _k}} \right)} \right) \le \mathpzc{R}} \right),}&{{\rm{Type~I}}}\\
{\Pr \left( {{{\log }_2}\left( {1 + \sum\limits_{l = 1}^k {{\gamma _l}} } \right) \le \mathpzc{R}} \right),}&{{\rm{CC}}}\\
{\Pr \left( {\sum\limits_{l = 1}^k {{{\log }_2}\left( {1 + {\gamma _l}} \right)}  \le \mathpzc{R}} \right),}&{{\rm{IR}}.}
\end{array}} \right.
\end{multline}
where $\mathpzc{R}=b/L$ is the transmission rate (normalized reward per renewal), $\gamma_l = \gamma_T \alpha_l$ represents the signal-to-noise ratio (SNR) of the $l$-th transmission, $\gamma_T$ and $\alpha_l$ correspond to the transmit SNR and the $l$-th channel gain, respectively and $p_0=1$ by convention. The outage performance of HARQ schemes under various fading channels has been extensively investigated. In particular, the outage probabilities of Type I HARQ and HARQ-CC have been derived in closed-form by considering general fading channels in \cite[eqs. (8), (11)]{shi2016optimal}. Moreover, the outage expression of HARQ-IR has been given in \cite[eq. (17)]{shi2017asymptotic}. Nonetheless, the specific expressions are omitted here due to space limitation.
\subsubsection{Maximum Arrival Rate}\label{sec:max_arr_rate}
In \cite{li2015throughput}, the effective capacity is used to assess the maximum arrival rate of HARQ schemes. To do so, by putting \eqref{eqn:outage_definition} into \eqref{eqn:q_k_harq} and then combining with \eqref{eqn:c_e_fin} and \eqref{eqn:z_1_root_eq}, the effective capacities of the three HARQ schemes can be calculated as ${C_e} = {{\ln \zeta }}/{\theta }$, where $\sum\nolimits_{k = 1}^{\mathcal K} {{q_k}{\zeta ^{kL}}}  = {e^{\theta b}}$. By introducing $\bar \zeta = {\zeta ^L} $ and $\bar \theta = L\theta$, the effective capacity can be obtained as
\begin{equation}\label{eqn:eff_con_rate_harq}
  {C_e} 
  = \frac{{\ln \bar \zeta }}{{\bar \theta }},
\end{equation}
where $\sum\nolimits_{k = 1}^{\mathcal K} {{q_k}{{\bar \zeta }^k}}  = {e^{\bar \theta \mathpzc{R}}}$. From \eqref{eqn:eff_con_rate_harq}, the effective capacity of the HARQ scheme is equivalent to that of a new renewal service process with the interarrival time ranged from 1 to ${\mathcal K}$ and the constant reward $\mathpzc{R}$. Moreover, $C_e$ can be approximated for small $\theta$ by using the Taylor series expansion as illustrated in Subsection \ref{sec:app_ce_con}. With \eqref{eqn:ce_bounds_con}, the effective capacity is bounded as $R/K\le C_e \le {{\mathpzc{R}}}/{{\sum\nolimits_{k = 0}^{{\mathcal K} - 1} {{p_k}} }}$ and $\mathop {\lim }\nolimits_{\theta  \to \infty} C_e={\mathpzc R}/{K}$.

Since both the decoding successes and failures are counted as rewards, the effective capacity obtained in \eqref{eqn:eff_con_rate_harq} indicates the maximum service rate supported by HARQ systems. In addition, the maximum arrival rate can be used to evaluate the throughput of the lossless HARQ scheme, which is allowed to have an infinite number of transmissions, i.e., $\mathcal K = \infty$, to guarantee no decoding failures.


\subsubsection{Outage Effective Capacity}
As pointed out in Subsection \ref{sec:max_arr_rate}, the maximum arrival rate does not imply that all the incoming data will be successfully delivered to the receiver because of the presence of decoding failures. In particular, the outage events usually take place for the truncated HARQ schemes, i.e. ${\mathcal K}<\infty$. To address this issue, the outage effective capacity is proposed in \cite{qiao2016outage,larsson2016effective}. Unlike the maximum arrival rate studied in Subsection \ref{sec:max_arr_rate}, the reward of each renewal is a two-value function for the truncated HARQ schemes. More specifically, the reward is $b$ if the message is successfully decoded and zero otherwise. Therefore, the analytical results in Section \ref{sec:ec_gen} can be used herein to evaluate the outage effective capacity. For truncated HARQ schemes, $X_i$ and $\frak S_i$ represent the number of HARQ transmissions involved into the delivery of the $i$-th message, and the termination state whether the receiver successfully recovers the message or not, respectively. Hence, the joint pmf $q_{k,s}$ of truncated HARQ schemes is given by
\begin{align}\label{eqn:q_trun_harq}
&\Pr(X = k L,\frak S=s) ={q_{k,s}}\notag\\
&\quad \quad \quad = \left\{ {\begin{array}{*{20}{c}}
{{p_{k - 1}} - {p_k},}&{k \le {\mathcal K}\& s = \rm S}\\
{{p_{\mathcal K}},}&{k = {\mathcal K}\& s = \rm F}
\end{array}} \right.,
\end{align}
where the notations $\mathrm S$ and $\mathrm F$ denote the success and the failure of the decoding, respectively. The reward function $\mathcal R(X,\frak S)$ is explicitly given by
\begin{align}\label{eqn:rew_trun_harq}
&\mathcal R(X = k L,\frak S=s) = \left\{ {\begin{array}{*{20}{c}}
{b,}&{k \le {\mathcal K}\& s = \rm S}\\
{0,}&{k = {\mathcal K}\& s = \rm F}
\end{array}} \right..
\end{align}
By using \eqref{eqn:ce_var_limit}, the outage effective capacity can be obtained as ${C_{e}^{out}} = {{\ln {\zeta} }}/{\theta }$, where $\zeta$ satisfies $\sum\nolimits_{k = 1}^{{\mathcal K}} {{q_{k,\rm S}}{e^{b\theta }}{\zeta ^{kL}}}  + {q_{{\mathcal K},\rm F}}{\zeta ^{{\mathcal K}L}} = 1$. Similarly, with the same definitions of $\bar \zeta = {\zeta ^L} $ and $\bar \theta = L\theta$, the outage effective capacity can be rewritten as ${C_{e}^{out}} = {{\ln {\bar \zeta} }}/{\bar \theta }$, where $\bar \zeta$ satisfies 
\begin{equation}\label{eq:g_tilde_trunharq}
 \sum\limits_{k = 1}^{{\mathcal K}} {{q_{k,\rm S}}{e^{\mathpzc{R}\bar \theta }}{\bar\zeta ^k}}  + {q_{{\mathcal K},\rm F}}{\bar \zeta ^{\mathcal K}} = 1,\, \bar \zeta \ge 0.
\end{equation}
Accordingly, \eqref{eq:g_tilde_trunharq} implies that ${C_{e}^{out}}$ is equivalent to the outage effective capacity of a simplified HARQ service process with the interarrival time ranged from 1 to ${\mathcal K}$, the QoS exponent $\bar \theta$ and the constant reward $\mathpzc{R}$. This result is consistent with \cite{larsson2016effective}. By using \eqref{eqn:ce_app_var}, the outage effective capacity can be approximated for small $\theta$. Moreover, from \eqref{eqn:ec_bound_fin_var}, since $({\mathcal K}L,{\rm F}) = \mathop {\arg }\nolimits_{\left( {\kappa ,j} \right)} \min\nolimits_{q_{\kappa ,j}>0} {\mathcal R(X,\frak S)}$, the outage effective capacity is bounded as
\begin{equation}\label{eqn:out_eff_cap_bounds}
0 \le {{C_{e}^{out}}} \le \min \left\{ { - \frac{{\ln {q_{{\mathcal K},\rm F}}}}{{{\mathcal K} \bar \theta }},\frac{{\mathbb E\left( {{\cal R}(X,\frak S)} \right)}}{{\mathbb E\left( X \right)}}} \right\},
\end{equation}
where the LTAT ${{\mathbb E\left( {{\cal R}(X,\frak S)} \right)}}/{{\mathbb E\left( X \right)}}$ is given by\cite{chelli2014performance}
\begin{equation}\label{eqn:ltat}
\frac{{\mathbb E\left( {{\cal R}(X,\frak S)} \right)}}{{\mathbb E\left( X \right)}} = \frac{{\mathpzc{R}\left( {1 - {p_{\mathcal K}}} \right)}}{{\sum\nolimits_{k = 0}^{{\mathcal K} - 1} {{p_k}} }}.
\end{equation}
As expected, the outage effective capacity is less than or equal to the LTAT due to the constraint of the limited buffer length. In addition, \eqref{eqn:ec_lower_bounds_lbd} indicates that
\begin{equation}\label{eqn:ec_}
\mathop {\lim }\limits_{\theta  \to  \infty}{{C_{e}^{out}}} = 0.
\end{equation}
This is different from the maximal arrival rate because erroneously received messages are counted as null rewards. As $\theta$ increases, the buffer overflow probability becomes much more stringent, which consequently leads to the continuous declination of the outage effective capacity.
\subsection{Variable-Rate HARQ-IR Systems}
Unlike the conventional FR-HARQ schemes, the transmission rate of the HARQ-IR scheme could be changeable from one transmission to another. More specifically, we assume that the length of the $k$-th subcodeword is $L_k$. By applying capacity achieving channel coding, an outage event happens when the accumulated mutual information is below $b$. The outage probability after $k$ HARQ rounds is thus written as \cite{szczecinski2013rate}
\begin{align}\label{eqn:out_prob_def}
{\hat p_{k}}
&={\rm{Pr}}\left( {\sum\limits_{l = 1}^k {\frac{1}{\mathpzc{R}_l}{{\log }_2}\left( {1 + \gamma_l} \right)}  < 1}
 \right),
\end{align}
where $\mathpzc{R}_k = b/L_k$ for $k\in [1,{\mathcal K}]$. From \eqref{eqn:out_prob_def}, the outage probability of VR-HARQ scheme can be derived in closed-form if the fading channels are independently Nakagami-m distributed among different HARQ rounds. The channel gains $\alpha_1,\cdots,\alpha_{\mathcal K}$ are independent Gamma random variables under the circumstance, and we assume $\alpha_l \sim \mathcal G(m_l,\Omega_l/m_l)$, where $m_l$ and $\Omega_l$ stand for the fading order and average channel power gain, respectively. As proved in Appendix \ref{app:pk_hat}, by using Mellin transform, ${\hat p_{k}}$ can be derived in terms of the generalized Fox's H function as \eqref{eqn:pk_F_G_foxh}, shown at the top of the next page, where the explicit definition of the generalized Fox's H function $Y_{p,q}^{m,n}(\cdot)$ is omitted here to conserve space (see \cite{yilmaz2010outage}).
\begin{figure*}[!t]
\begin{equation}\label{eqn:pk_F_G_foxh}
{\hat p_{k}} = Y_{1,k + 1}^{k,1}\left[ {\left. {\begin{array}{*{20}{c}}
{\left( {1,1,0,1} \right)}\\
{\left( {1,\frac{1}{{{\mathpzc R_1}}},\frac{{{m_1}}}{{{\gamma _T}{\Omega _1}}},{m_1}} \right), \cdots ,\left( {1,\frac{1}{{{\mathpzc R_k}}},\frac{{{m_k}}}{{{\gamma _T}{\Omega _k}}},{m_k}} \right),\left( {0,1,0,1} \right)}
\end{array}} \right|2\prod\limits_{l = 1}^k {{{\left( {\frac{{{m_l}}}{{{\gamma _T}{\Omega _l}}}} \right)}^{\frac{1}{{{\mathpzc R_l}}}}}} } \right],
\end{equation}
\hrulefill
\end{figure*}

The joint pmf of the interarrival time $\hat X$ and the termination state $\hat{\frak S}$, ${{\hat q}_{k,s}}$, is given by
\begin{align}\label{eqn:q_hat_jointpmf}
&\Pr \left( {\hat X = \sum\nolimits_{l = 1}^k {{L_l}} ,\hat {\frak S} = s} \right) = {{\hat q}_{k,s}}\notag\\
& \quad \quad \quad = \left\{ {\begin{array}{*{20}{c}}
{{{\hat p}_{k - 1}} - {{\hat p}_k},}&{k \le {\mathcal K}\& s = {\rm{S}}}\\
{{{\hat p}_{\mathcal K}},}&{k = {\mathcal K}\& s = {\rm{F}}}
\end{array}} \right..
\end{align}
Moreover, the corresponding reward function $\mathcal R(\hat X,\hat{\frak S})$ is expressed as
\begin{align}\label{eqn:rew_trun_harq}
\mathcal R(\hat X =  \sum\nolimits_{l = 1}^k {{L_l}} ,\hat{\frak S}=s) 
& = \left\{ {\begin{array}{*{20}{c}}
{b,}&{k \le {\mathcal K}\& s = \rm S}\\
{0,}&{k = {\mathcal K}\& s = \rm F}
\end{array}} \right..
\end{align}

Therefore, the effective capacity is obtained as $C_e^{out} = {{\ln  \zeta }}/{{ \theta }} $, where $ \zeta$ is given by $\sum\nolimits_{k = 1}^{\mathcal K} {{{\hat q}_{k,{\rm{S}}}}{e^{b\theta }}{\zeta ^{\sum\nolimits_{l = 1}^k {{L_l}} }}}  + {{\hat q}_{{\mathcal K},{\rm{F}}}}{\zeta ^{\sum\nolimits_{l = 1}^{\mathcal K} {{L_l}} }} = 1$. By defining $\hat \theta  = {b \theta }$ and $\hat \zeta = {{{ \zeta }^b}}$, the effective capacity can be rewritten as $C_e = {{\ln \hat \zeta }}/{{\hat \theta }}$, where $\hat \zeta$ satisfies
\begin{equation}\label{eqn:zeta_tilde_ob}
\sum\limits_{k = 1}^{\mathcal K} {{{\hat q}_{k,{\rm{S}}}}{e^{\hat \theta }}{{\hat \zeta }^{\sum\limits_{l = 1}^k {{{{\mathpzc R_l}}^{-1}}} }}}  + {{\hat q}_{{\mathcal K},{\rm{F}}}}{{\hat \zeta }^{\sum\limits_{l = 1}^{\mathcal K} {{{{\mathpzc R_l}}^{-1}}} }} = 1.
\end{equation}
With \eqref{eqn:ce_app_var}, the effective capacity can be approximated.  In analogous to \eqref{eqn:out_eff_cap_bounds}, the effective capacity is bounded as
\begin{equation}\label{eqn:varout_eff_cap_bounds}
0 \le {{C_{e}}} \le \min \left\{ { - \frac{{\ln {\hat q_{{\mathcal K},\rm F}}}}{{\hat \theta\sum\nolimits_{l=1}^{\mathcal K} {\mathpzc R_l}^{-1}  }},\frac{{\mathbb E\left( {{\cal R}(\hat X,\hat{\frak S})} \right)}}{{\mathbb E\left( \hat X \right)}}} \right\},
\end{equation}
where the LTAT ${{\mathbb E\left( {{\cal R}(\hat X,\hat {\frak S})} \right)}}/{{\mathbb E\left( \hat X \right)}}$ is given by \cite{szczecinski2013rate}
\begin{equation}\label{eqn:var_ltat}
\frac{{\mathbb E\left( {{\cal R}(\hat {X},\hat{\frak S})} \right)}}{{\mathbb E\left( \hat X \right)}} = \frac{{{1 - {\hat p_{\mathcal K}}} }}{{\sum\limits_{k = 0}^{{\mathcal K} - 1} {{\mathpzc R_{k+1}}^{-1}\hat p_k} }}.
\end{equation}
Furthermore, the result $\mathop {\lim }\nolimits_{\theta  \to  \infty}{{C_{e}^{out}}} = 0$ similar to \eqref{eqn:ec_} can be obtained.

\subsection{Cross-Packet HARQ Systems}
In the retransmissions of both the FR- and VR-HARQ schemes, the retransmitted subcodewords do not involve new information bits. Particularly for FR-HARQ, the provisioning of the throughput close to ergodic capacity is prevented from exploiting the possible redundant mutual information. Instead, the XP-HARQ scheme was devised to avoid the waste of mutual information to substantially improve the throughput \cite{jabi2017adaptive}. More specifically, the XP-HARQ scheme introduces new information bits in retransmissions besides the redundancy bits. We assume the number of the new introduced information bits in the $k$-th HARQ round is $b_k$, and the number of the information bits in the initial transmission is $b_1$. During the $k$-th HARQ round, the currently introduced information bits is first concatenated with all the previously introduced ones to construct a long message. The resultant message thus contains $\sum\nolimits_{l=1}^k b_l$ information bits, and is then encoded into the $k$-th codeword of length $L$ symbols. According to the decoding conditions of XP-HARQ studied in \cite{jabi2017adaptive}, the outage event takes place after $k$ HARQ rounds if and only if the accumulated mutual information is below the number of the delivered information bits in the current and previous HARQ rounds. Therefore, the outage probability of the XP-HARQ scheme can be obtained as\cite{jabi2017adaptive}
\begin{align}\label{eqn:p_tilde_xpharq}
{{\breve p}_k} &= {\rm{Pr}}\left\{ {\bigcup\limits_{\kappa  = 1}^k {\left( {\sum\limits_{l = 1}^\kappa  {L{{\log }_2}\left( {1 + {\gamma _l}} \right)}  < \sum\limits_{l = 1}^\kappa  {{b_l}} } \right)} } \right\}\notag\\
 & = {\rm{Pr}}\left\{ {\bigcup\limits_{\kappa  = 1}^k {\left( {\underbrace {\sum\limits_{l = 1}^\kappa  {{{\log }_2}\left( {1 + {\gamma _l}} \right)} }_{{\mathcal I_\kappa }} < \sum\limits_{l = 1}^\kappa  {{{\breve {\mathpzc R}}_l}} } \right)} } \right\},
\end{align}
where ${{\breve {\mathpzc R}}_l} = b_l/L$ and $\bigcup(\cdot)$ represents the union of events. Unfortunately, due to the presence of the correlation among the accumulated mutual informations per symbol ${\mathcal I_1 },\cdots,{\mathcal I_{\mathcal K} }$, it is intractable to derive an exact expression for \eqref{eqn:p_tilde_xpharq} and there is still no
readily available result in the literature. In this paper, ${{\breve p}_k}$ is computed by conducting Monte Carlo simulations.

In order to obtain the effective capacity of the XP-HARQ scheme, the pmf of the interarrival time $\breve X$ and the termination state $\breve{\frak S}$, ${{\breve q}_{k,s}}$, is derived as
\begin{align}\label{eqn:q_tilde_jointpmf}
&\Pr \left( {\breve X = kL ,\breve {\frak S} = s} \right) = {{\breve q}_{k,s}}\notag\\
& \quad \quad \quad = \left\{ {\begin{array}{*{20}{c}}
{{{\breve p}_{k - 1}} - {{\breve p}_k},}&{k \le {\mathcal K}\& s = {\rm{S}}}\\
{ {{\breve p}_{\mathcal K}},}&{k = {\mathcal K}\& s = {\rm{F}}}
\end{array}} \right..
\end{align}
And the reward function $\mathcal R(\breve X,\breve{\frak S})$ is given by
\begin{align}\label{eqn:rew_xp_harq}
\mathcal R(\breve X =  kL ,\breve{\frak S}=s) 
& = \left\{ {\begin{array}{*{20}{c}}
{\sum\nolimits_{l=1}^k b_l,}&{k \le {\mathcal K}\& s = \rm S}\\
{0,}&{k = {\mathcal K}\& s = \rm F}
\end{array}} \right..
\end{align}

Accordingly, the effective capacity of XP-HARQ is calculated as $C_e^{out} = {{\ln  \zeta }}/{{ \theta }} $, where $ \zeta$ is determined by $\sum\nolimits_{k = 1}^{\mathcal K} {{{\breve q}_{k,{\rm{S}}}}{e^{\theta \sum\nolimits_{l = 1}^k {{b_l}} }}{\zeta ^{kL}}}  + {{\breve q}_{{\mathcal K},{\rm{F}}}}{\zeta ^{kL}} = 1$. By defining $\breve \theta  = {L \theta }$ and $\breve \zeta = {{{ \zeta }^L}}$, we have $C_e = {{\ln \breve \zeta }}/{{ \breve \theta }} $, where $\breve \theta$ is given by
\begin{equation}\label{eqn:eff_xp_harq_tilde_zeta}
\sum\limits_{k = 1}^{\mathcal K} {{{\breve q}_{k,{\rm{S}}}}{e^{\breve \theta \sum\limits_{l = 1}^k {{{\breve {\mathpzc R}}_l}} }}{{\breve \zeta }^k}}  + {{\breve q}_{{\mathcal K},{\rm{F}}}}{{\breve \zeta }^k} =  1.
\end{equation}
With \eqref{eqn:ce_app_var}, the effective capacity can be approximated.  In analogous to \eqref{eqn:out_eff_cap_bounds}, the effective capacity is bounded as
\begin{equation}\label{eqn:varout_eff_cap_boundstilde}
0 \le {{C_{e}}} \le \min \left\{ { - \frac{{\ln {\breve q_{{\mathcal K},\rm F}}}}{{{\mathcal K}\breve \theta  }},\frac{{\mathbb E\left( {{\cal R}(\breve X,\breve{\frak S})} \right)}}{{\mathbb E\left( \breve X \right)}}} \right\},
\end{equation}
where the LTAT of XP-HARQ, ${{\mathbb E\left( {{\cal R}(\breve X,\breve {\frak S})} \right)}}/{{\mathbb E\left( \breve X \right)}}$, is given by \cite{jabi2017adaptive}
\begin{equation}\label{eqn:var_ltat_tilde}
\frac{{\mathbb E\left( {{\cal R}(\breve {X},\breve{\frak S})} \right)}}{{\mathbb E\left( \breve X \right)}} = \frac{{\sum\nolimits_{k = 1}^{\mathcal K} {\breve {\mathpzc R_k}\left( {{\breve p_{k - 1}} - {\breve p_{\mathcal K}}} \right)} }}{{\sum\nolimits_{k = 0}^{{\mathcal K} - 1} {{\breve p_k}} }}.
\end{equation}
Furthermore, the similar result to \eqref{eqn:ec_} can be demonstrated, i.e., $\mathop {\lim }\nolimits_{\theta  \to  \infty}{{C_{e}}} = 0$.

Finally, the calculations of the effective capacity for various HARQ schemes are briefly summarized in Table \ref{tab:eff_cap}.
\begin{figure*}[!t]
\begin{center}
\captionsetup{type=table} 
\caption{The effective capacity for various HARQ schemes.}
\label{tab:eff_cap}%
\begin{tabular}{|p{3.2cm}||p{1.8cm}|p{1.1cm}|p{5.6cm}|p{2.3cm}|}
 \hline
 \diagbox[width=3.6cm]{\textbf{HARQ}}{\textbf{Metrics}}& $C_e$ & $\theta$ & $\zeta$ & $\frac{{\mathbb E\left( {{\cal R}(X,\frak S)} \right)}}{{\mathbb E\left( X \right)}}$\\
 \hline
 \hline
 FR-HARQ-Maximum Arrival Rate  & \multirow{4}*{${C_e}  = {{\ln \bar \zeta }}/{{\bar \theta }}$}    & \multirow{4}*{$ \bar \theta = L \theta$}&   $\sum\limits_{k = 1}^{\mathcal K} {{q_k}{{\bar \zeta }^k}}  = {e^{\bar \theta \mathpzc{R}}}$ & $\frac{{\mathpzc{R}}}{{\sum\limits_{k = 0}^{{\mathcal K} - 1} {{p_k}} }}$\\
 \cline{1-1}\cline{4-5}
 FR-HARQ-Outage Effective Capacity &   ~  &   ~ &$\sum\limits_{k = 1}^{{\mathcal K}} {{q_{k,\rm S}}{e^{\mathpzc{R}\bar \theta }}{\bar\zeta ^k}}  + {q_{{\mathcal K},\rm F}}{\bar \zeta ^{\mathcal K}} = 1$&$\frac{{\mathpzc{R}\left( {1 - {p_{\mathcal K}}} \right)}}{{\sum\limits_{k = 0}^{{\mathcal K} - 1} {{p_k}} }}$\\
 \hline
 VR-HARQ &$C_e = {{\ln \hat \zeta }}/{{\hat \theta }}$& $\hat \theta = b\theta$ &  $\sum\limits_{k = 1}^{\mathcal K} {{{\hat q}_{k,{\rm{S}}}}{e^{\hat \theta }}{{\hat \zeta }^{\sum\limits_{l = 1}^k {{{{\mathpzc R_l}}^{-1}}} }}}  + {{\hat q}_{{\mathcal K},{\rm{F}}}}{{\hat \zeta }^{\sum\limits_{l = 1}^{\mathcal K} {{{{\mathpzc R_l}}^{-1}}} }} = 1$ & $\frac{{{1 - {\hat p_{\mathcal K}}} }}{{\sum\limits_{k = 0}^{{\mathcal K} - 1} {{\mathpzc R_{k+1}}^{-1}\hat p_k} }}$\\
 \hline
 XP-HARQ &$C_e = {{\ln \breve \zeta }}/{{ \breve \theta }}$ & $\breve \theta  = {L \theta }$ &  $\sum\limits_{k = 1}^{\mathcal K} {{{\breve q}_{k,{\rm{S}}}}{e^{\breve \theta \sum\limits_{l = 1}^k {{{\breve {\mathpzc R}}_l}} }}{{\breve \zeta }^k}}  + {{\breve q}_{{\mathcal K},{\rm{F}}}}{{\breve \zeta }^k} =  1$&$\frac{{\sum\limits_{k = 1}^{\mathcal K} {\breve {\mathpzc R_k}\left( {{\breve p_{k - 1}} - {\breve p_{\mathcal K}}} \right)} }}{{\sum\limits_{k = 0}^{{\mathcal K} - 1} {{\breve p_k}} }}$\\
 \hline
\end{tabular}
\end{center}
\end{figure*}

\section{Verifications and Discussions}\label{sec:ver}
In this section, numerical examples are presented for verifications and discussions. Unless otherwise specified, we set $\gamma_T=20$dB, $\mathpzc R=4$bps/Hz, ${\mathcal K}=5$ and $b=1080$bits \cite{liu2004cross}. In addition, we assume that all HARQ links experience independent Rayleigh fading with unit average power, i.e., ${\mathbb E}(\alpha_l)=\Omega_l=1$ and $m_l=1$ for $l \in [1,\mathcal K]$. 
\subsection{FR-HARQ Scheme}
 To start with, the maximum arrival rate against the QoS exponent $\theta$ for the three conventional FR-HARQ schemes is depicted in Fig. \ref{fig:theta}, and the approximate results obtained in \cite{li2015throughput} are provided for comparison. It is observed that $C_e$ decreases with $\theta$, which is consistent with our analysis. This is because that higher $\theta$ represents stricter queuing constraints imposed on buffer overflow probability, which limits the increase of service arrival rate, and consequently results in the decrease of the capacity. As shown in Fig. \ref{fig:theta}, the gap between the approximate and the exact results grows with $\theta$, which justifies the significance of the exact analysis, especially, for large $\theta$. As $\theta$ approaches to infinity, the maximum arrival rate tends to a lower bound that is given by \eqref{eqn:ce_bounds_con}, i.e., $\mathop {\lim }\nolimits_{\theta  \to \infty}  C_e = R/K = 0.8$bps/Hz.  Moreover, Fig. \ref{fig:theta} further substantiates that HARQ-IR is superior to other HARQ schemes in terms of the maximum arrival rate.

 In Fig. \ref{fig:outage_theta}, the outage effective capacity $C_e$ is plotted versus the QoS exponent. It can be seen that the tendency of $C_e$ for the three conventional HARQ schemes is the same as Fig. \ref{fig:theta}, and their approximations given by \eqref{eqn:ce_app_var} are also shown for comparison. It is readily found that the approximate results coincide with the exact ones under small QoS exponent. Unlike the maximum arrival rate in Fig. \ref{fig:theta}, the outage effective capacity approaches to zero as $\theta$ tends to infinity. This difference is due to the fact that the amount of unsuccessfully delivered data is not counted as reward while computing the outage effective capacity. The smaller $\theta$ means the stricter QoS constraint, which needs smaller arrival rate $\mu$ to ensure a tighter constraint on buffer overflow probability. 

\begin{figure}
  \centering
  \includegraphics[width=3.2in,height=1.6in]{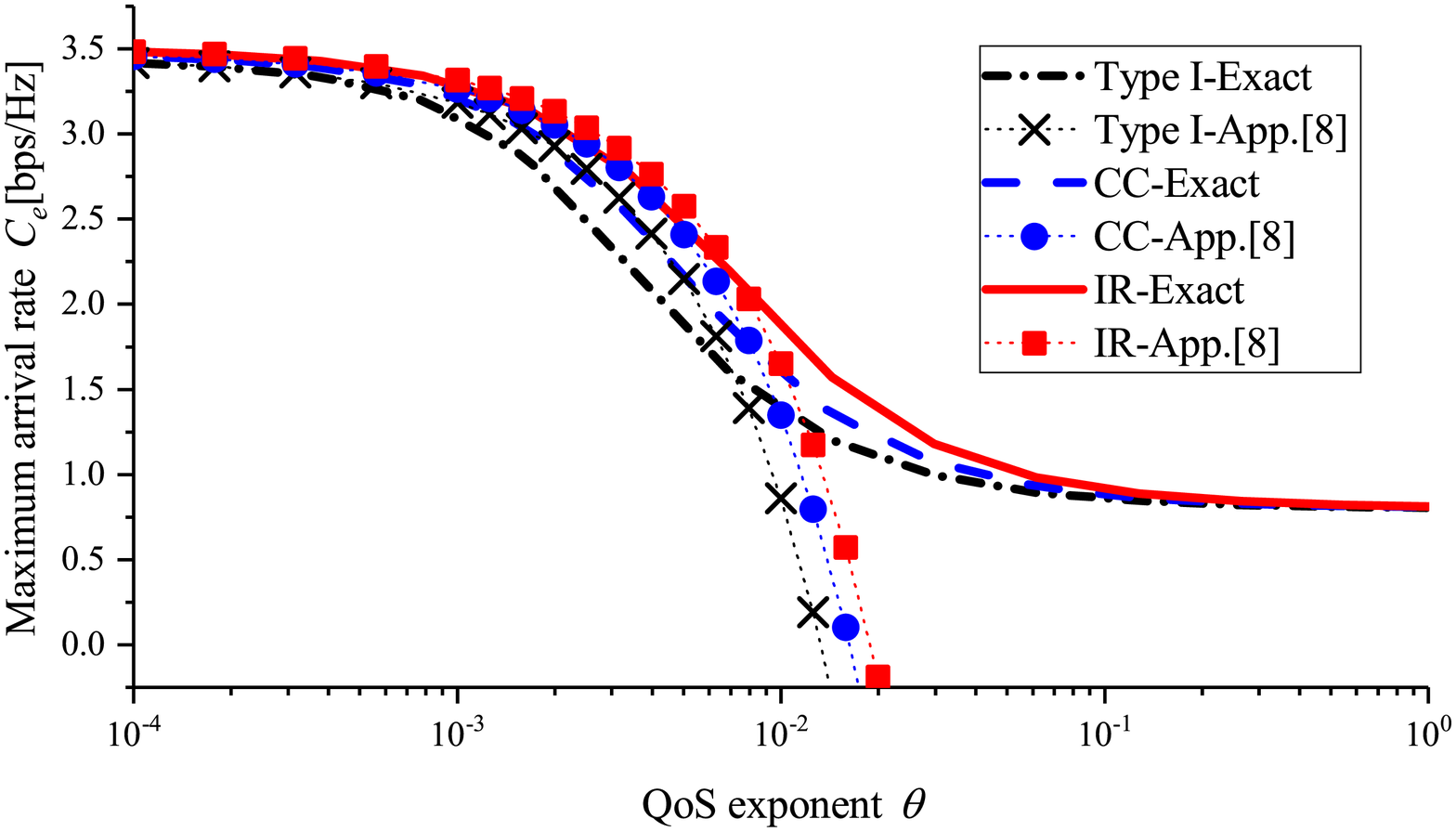}
  \caption{The maximum arrival rate versus the QoS exponent for FR-HARQ schemes.}\label{fig:theta}
\end{figure}

\begin{figure}
  \centering
  \includegraphics[width=3.2in,height=1.6in]{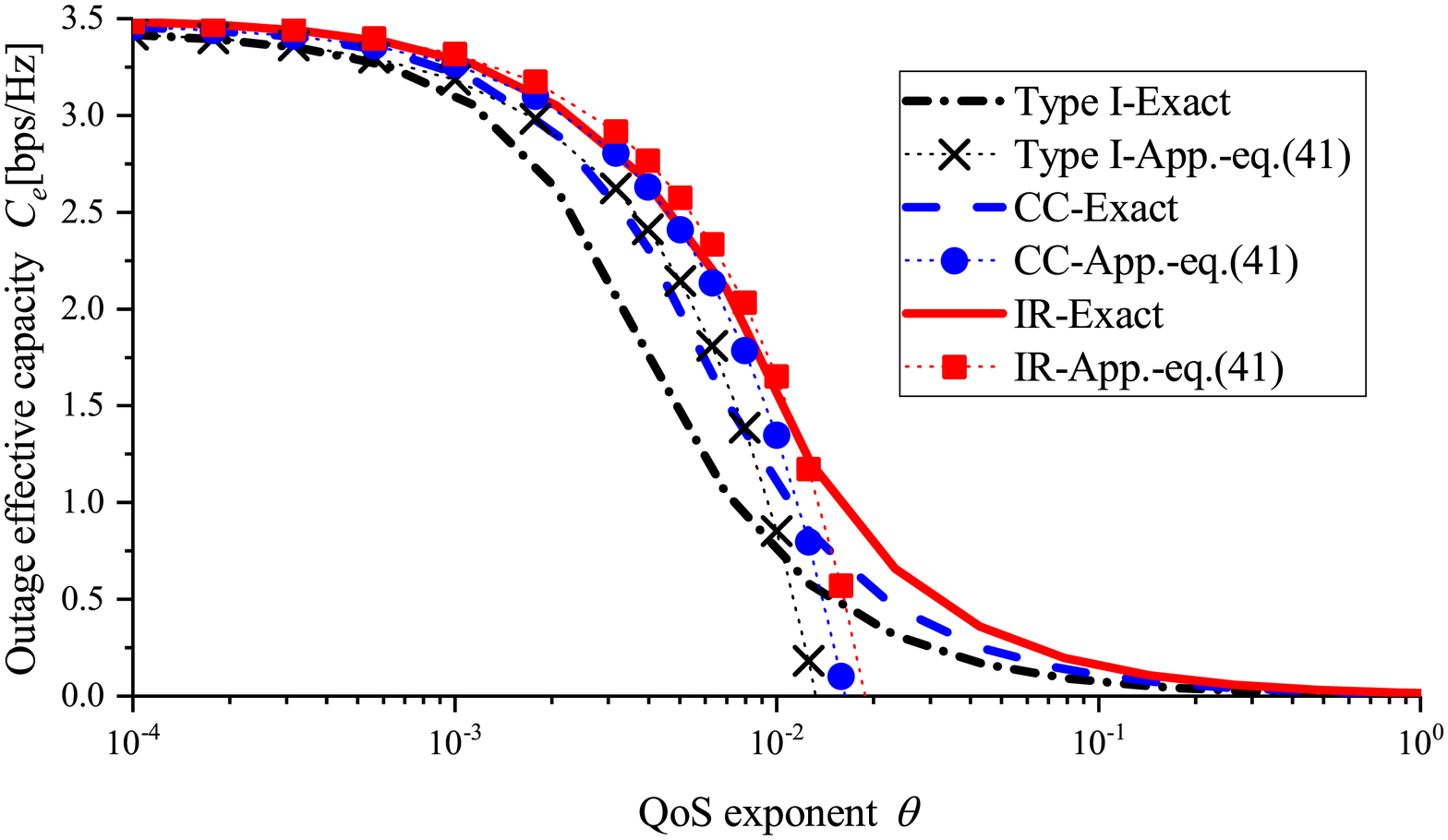}
  \caption{The outage effective capacity versus the QoS exponent for FR-HARQ schemes.}\label{fig:outage_theta}
\end{figure}

Fig. \ref{fig:outage_snr} illustrates the impact of the transmit SNR $\gamma_T$ on the effective capacity, and also shows the comparison between the maximum arrival rate and the outage effective capacity. It is not out of expectation that both the maximum arrival rate and the effective capacity increase with $\gamma_T$. Nonetheless, both of them are found to be bounded no matter how high $\gamma_T$ is, because the bounds of the maximum arrival rate given in \eqref{eqn:ce_bounds_con} manifest that $0.8$bps/Hz $\le C_e \le {\mathpzc{R}}/{\mathbb E(X)} < \mathpzc{R}=4$bps/Hz. Whereas, the bounds of the outage effective capacity given by \eqref{eqn:ec_bound_fin_var} show that $0$bps/Hz $\le C_e \le {\mathpzc{R}\left( {1 - {p_{\mathcal K}}} \right)}/{\mathbb E(X)} < \mathpzc{R}=4$bps/Hz.
\begin{figure}
  \centering
  \includegraphics[width=3.2in,height=1.6in]{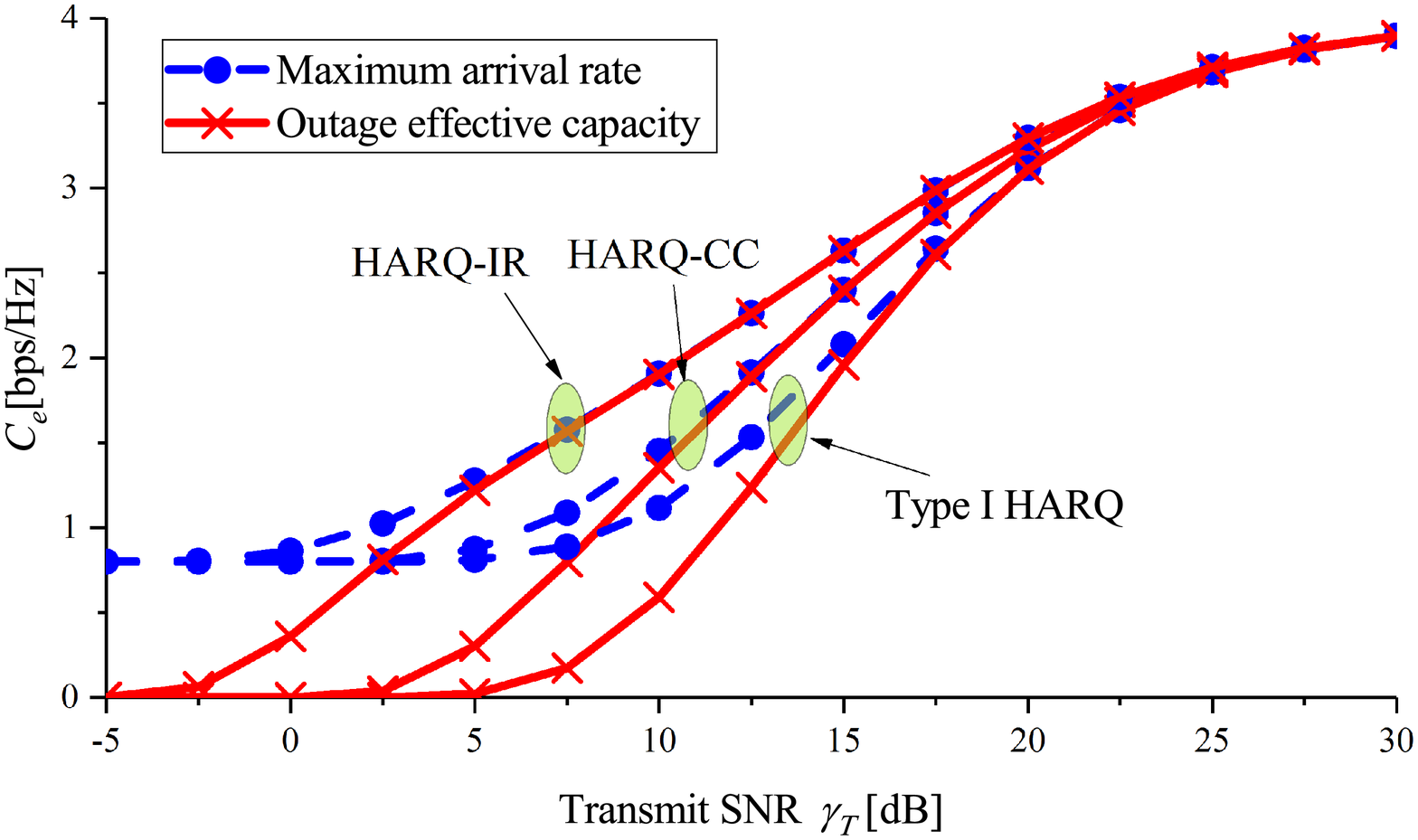}
  \caption{The comparison between the maximum arrival rate and the outage effective capacity with parameter $\theta=10^{-3}$.}\label{fig:outage_snr}
\end{figure}

\subsection{VR-HARQ Scheme}
Fig. \ref{fig:vr_theta} shows the effective capacity of the VR-HARQ schemes against the effect of the QoS exponent, and the notation ${\bf r}$ herein is defined as ${\bf r} = [\mathpzc R_1,\cdots,\mathpzc R_K]$. It is observed in Fig. \ref{fig:vr_theta} that the exact results perfectly agree with the approximate ones for small $\theta$. Similarly to Fig. \ref{fig:outage_theta}, the effective capacity decreases to zero as $\theta$ increases to infinity in Fig. \ref{fig:vr_theta}. Unsurprisingly, the increase of the transmit SNR improves the effective capacity. For example, for fixed values of ${\bf r} = [4,3,3,2,2]$bps/Hz and $\theta=10^{-4}$, the effective capacity increases by 1.5bps/Hz if the transmit SNR is increased from $10$dB to $20$dB. Furthermore, it is seen from this figure that the rate selection has a significant impact on the effective capacity.
\begin{figure}
  \centering
  \includegraphics[width=3.2in,height=1.6in]{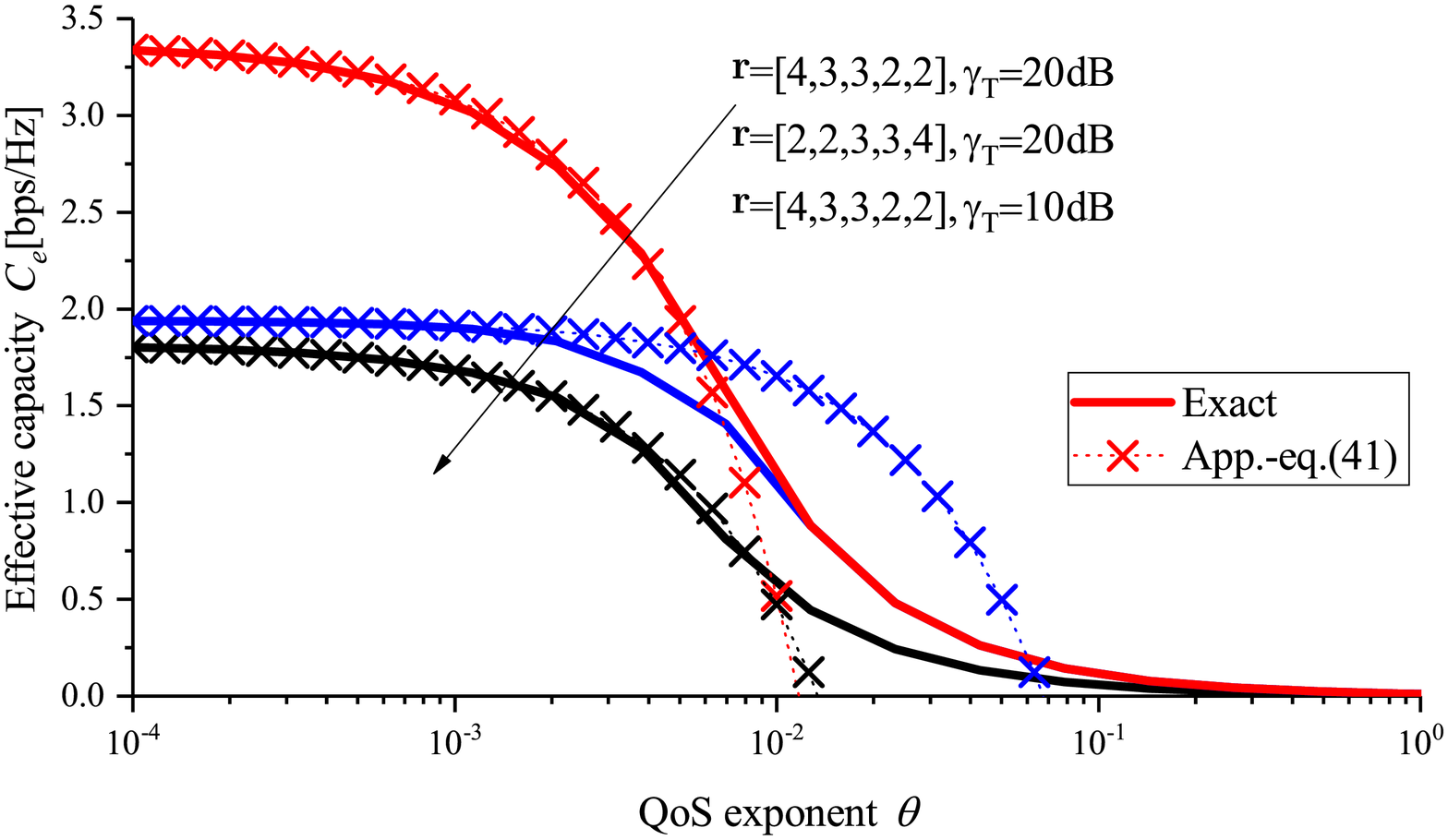}
  \caption{The effective capacity versus the QoS exponent for the VR-HARQ schemes.}\label{fig:vr_theta}
\end{figure}


\subsection{XP-HARQ Scheme}
The effective capacity for XP-HARQ scheme is plotted against the QoS exponent in Fig. \ref{fig:xp_snr}, wherein $\breve{\bf r} = [\breve {\mathpzc R}_1,\cdots,\breve{\mathpzc R}_K]$. It can be observed that there is a perfect agreement between the exact results and the approximate ones under a small $\theta$. Similarly, the effective capacity of the XP-HARQ scheme is a decreasing function of the QoS exponent as well as the transmit SNR. In addition, it is seen from Fig. \ref{fig:xp_snr} that the transmission rates of the XP-HARQ considerably influences the effective capacity.

\begin{figure}
  \centering
  \includegraphics[width=3.2in,height=1.6in]{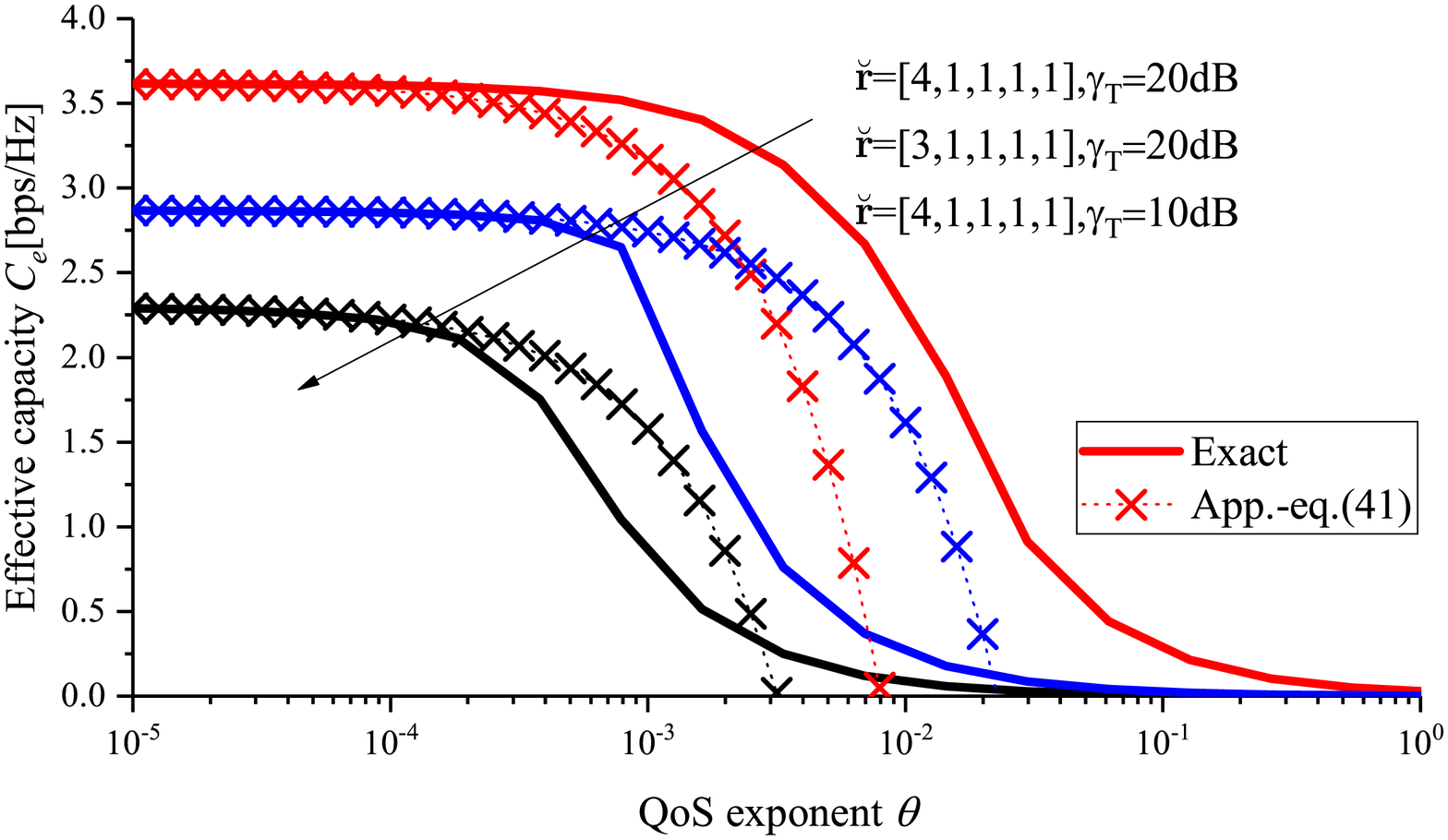}
  \caption{The effective capacity versus the QoS exponent for the XP-HARQ scheme.}\label{fig:xp_snr}
\end{figure}

Note that the selection of transmission rates has a critical impact on the effective capacity, the effective capacity can be maximized via optimizing transmission rates of the HARQ-IR scheme. To compare the three different HARQ-IR schemes, including FR-HARQ-IR, VR-HARQ and XP-HARQ, Fig. \ref{fig:com} exhibits their optimal effective capacities versus the transmit SNR. Nevertheless, the optimal design of transmission rates for HARQ schemes is out of the scope of this paper due to the non-convexity of the objective function. Similarly to \cite{jabi2017adaptive}, we conduct an exhaustive search over the space of following available transmission rates: FR-HARQ-IR and VR-HARQ adopt ${\mathpzc R}, {\mathpzc R}_k \in \{1.5,1.75,2,\cdots,3.75\}$bps/Hz, while XP-HARQ adopts $\breve{\mathpzc R}_1 \in \{1.5,1.75,2,\cdots,3.75\}$bps/Hz for the initial round and $\breve{\mathpzc R}_k \in \{0,0.25,0.5,\cdots,3.75\}$bps/Hz for the others. Fig. \ref{fig:com} reveals that VR-HARQ and XP-HARQ can achieve almost the same optimal effective capacity, while FR-HARQ performs the worst. This is because that VR-HARQ turns to the variable-length coding to fully exploit the statistical information of fading channels. Whereas, XP-HARQ attempts to incorporate new information bits into retransmission to make the utmost of the possible redundant mutual information.

\begin{figure}
  \centering
  \includegraphics[width=3.2in,height=1.6in]{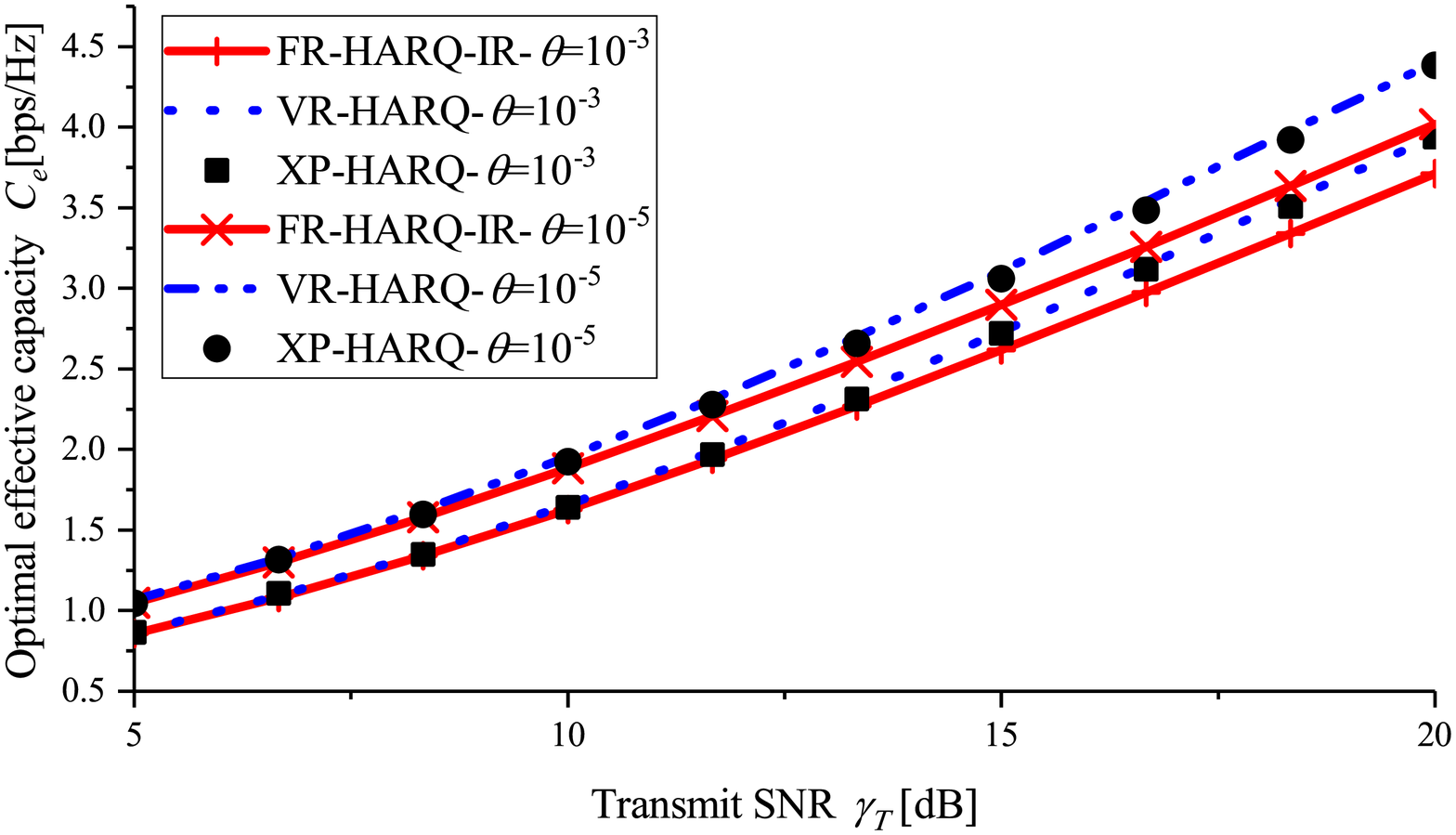}
  \caption{Optimal effective capacities of different HARQ-IR schemes with parameter $\mathcal K = 2$.}\label{fig:com}
\end{figure}

\section{Conclusions}\label{sec:con}
This paper has derived a unified exact formula for the effective capacity of the renewal network service process at a given QoS exponent, with which the cross-layer throughput for various HARQ systems has been accurately evaluated, including FR-HARQ (e.g., Type I HARQ, HARQ-CC and HARQ-IR), VR-HARQ and XP-HARQ. The formula not only has gained many insightful results, but also has paved the way for the exact cross-layer design by integrating the channel model of physical-layer to QoS requirements of link-layer. Specifically, the effective capacity has been found to decrease with QoS exponent as well as be bounded. Furthermore, if the transmission rates are optimally chosen to maximize the effective capacity, it has been shown that VR-HARQ and XP-HARQ achieve almost the same performance, and both of them surpass FR-HARQ.


\appendices
\section{Proof of Theorem \ref{the:inv_z_mat}}\label{app:inv_z_mat}
By using residue theorem, the inverse Z-transform of $\psi(z)$ can be expressed as
\begin{align}\label{eqn:mgf_nt_inv_sim}
{\cal Z}^{-1}\left\{ {{\psi}\left( z \right)} \right\}\left( t \right)
 &=  \sum\limits_{i = 1}^K
 {{\rm{Res}}\left\{ {\frac{{z^{t-1}\phi\left( z \right)}}{{\prod\limits_{k = 1}^K {\left( {z - {z_k}} \right)} }},z = {z_i}} \right\}}\notag\\
&=  \sum\limits_{i = 1}^K {\frac{{{z_i}^{t - 1}}\phi\left( z_i \right)}{{\prod\limits_{ j = 1, j \ne i}^K {\left( {z_i - {z_j}} \right)} }}},
\end{align}
where the notation of ${\rm Res}(f(x),u)$ stands for the residue of $f(x)$ at pole $x=u$. 
Identifying \eqref{eqn:mgf_nt_inv_sim} with the Laplace expansion of Vandermonde determinant yields
\begin{align}\label{eqn:mgf_rew_simpres_re}
 &{\cal Z}^{-1}\left\{ {{\psi}\left( z \right)} \right\}\left( t \right) \notag\\
 &= \frac{{\sum\limits_{i = 1}^K {{{\left( { - 1} \right)}^{K + i}}{{{z_i}^{t - 1}}\phi\left( z_i \right)}\prod\limits_{1 \le m \ne i < n \ne i \le K}^K {\left( {{z_n} - {z_m}} \right)} } }}{{\prod\limits_{1 \le m < n \le K}^K {\left( {{z_n} - {z_m}} \right)} }} \notag\\
 &= \frac{{\sum\limits_{i = 1}^K {{{{z_i}^{t - 1}}\phi\left( z_i \right)} {{{{F}}_{i,K}}} } }}{{\det {\bf{Z}}}}
 = \frac{{\det {\bf{F}} }}{{\det{\bf{Z}} }},
\end{align}
where ${{{{F}}_{i,j}}}$ represents the cofactor of the $(i,j)$-th entry of $\bf F$.

\section{Proof of Theorem \ref{the:dec_mon}}\label{app:dec_mon}
To complete the proof of Theorem \ref{the:dec_mon}, it will suffice to show that the first derivative of $\eta(x) = f(x)/x$ is larger than or equal to zero. The first derivative of $\eta(x)$ w.r.t. $x$ is
\begin{equation}\label{eqn:ce_first_der}
\eta^\prime(x) = \frac{1}{x }\left( {f'\left( x  \right) - \frac{{f\left( x  \right)}}{x }} \right).
\end{equation}


By repeatedly using mean value theorem, it follows that 
\begin{align}\label{eqn:ce_first_der_mvt}
\eta^\prime(x)
 &= \frac{1}{x }\left( {f^\prime \left( x  \right) - f^\prime({\alpha _1}x )} \right) \notag\\
 &= \left( {1 - {\alpha _1}} \right)f^{\prime \prime}\left( {{\alpha _1}x  + {\alpha _2}\left( {1 - {\alpha _1}} \right)x } \right).
\end{align}
where $\alpha_1,\alpha_2 \in (0,1)$. Clearly, the sign of $\eta^\prime(x)$ is determined by $f^{\prime \prime}(x)$. The monotonicity of $\eta(x)$ is thus proved.

\section{Bounds of $C_e$}\label{sec:app_prop}
Thanks to the monotonic decreasing function of $C_e$, $C_e$ is bounded as $\mathop {\lim }\nolimits_{\theta  \to \infty} {{u(\theta )}}/{\theta } \le {C_e} \le \mathop {\lim }\nolimits_{\theta  \to  0} {{u(\theta )}}/{\theta }$.
Clearly, the corresponding upper bound is $\mathop {\lim } \nolimits_{\theta  \to  0}{{u(\theta )}}/{\theta } = u^\prime (0)={R}/{{\mathbb E}(X)}$. With regard to the lower bound, \eqref{eqn:u_h_re} indicates ${e^{K\theta {C_e}}} \ge {e^{\theta R}} \ge {q_K}{e^{K\theta {C_e}}}$.
Accordingly, we have ${R}/{K} \le {C_e} \le {R}/{K} - {{\ln {q_K}}}/{{(K\theta) }}$. As $\theta$ approaches to infinity, using squeeze theorem yields $\mathop {\lim }\nolimits_{\theta  \to \infty} {{u(\theta )}}/{\theta }={R}/{K}$ if $q_K > 0$. The bounds of the effective capacity are consequently simplified into \eqref{eqn:ce_bounds_con}.
\section{Proof of \eqref{eqn:con_ec_ext}}\label{app:proof_cn_ec_ext}
Since $\mathbb E\left( {{{\zeta }^{ X}}} \right)$ can be approximated as
\begin{equation}\label{eqn:mgf_app}
\mathbb E\left( {{{\zeta }^{ X}}} \right)\approx \sum\limits_{k = 0}^K {{q_k}{\zeta ^{k\Delta x}}}  = \sum\limits_{k = 0}^K {{q_k}{{\left( {{\zeta ^{-\Delta x}}} \right)}^{-k}}}  = {\mathcal X_{{{\tilde X}_i}}}\left( z \right),
\end{equation}
where $z={\zeta ^{-\Delta x}}$ and the approximation becomes an equality as $\Delta \to 0$, the effective capacity of a renewal process by taking $\tilde X_i$ as the $i$-th interarrival time can be used to approximate $C_e$. Specifically, $C_e$ can be approximately obtained by using the same approach in Section \ref{sec:zt} as
\begin{align}\label{eqn:ce_app}
C_e &\approx - \mathop {\lim }\limits_{t \to \infty } \frac{1}{{\theta t}}\ln {\mathbb{E}}\left\{ {{e^{ - \theta  R{N_{t/\Delta x}}}}} \right\}\notag\\
&= \mathop {\lim }\limits_{t \to \infty } \frac{{\left( {\frac{t}{{\Delta x}} + K} \right)\ln {z}^{ - 1}}}{{\theta t}},
\end{align}
where $\sum\nolimits_{k = 0}^K {{q_k}{{z}^{ - k}}}  = {e^{\theta R}}$, i.e., ${\mathcal X_{{{\tilde X}_i}}}\left( z \right) = {e^{\theta R}}$. Due to the relationship between $\zeta$ and $z$, we have
\begin{align}\label{eqn:ce_app_rew}
C_e &= \mathop {\lim }\limits_{t \to \infty } \frac{{\left( {\frac{t}{{\Delta x}} + K} \right)\ln \zeta^{\Delta x}}}{{\theta_1 t}}=\frac{\ln\zeta}{\theta},
\end{align}
where $\mathbb E\left( {{{\zeta }^{ X}}} \right) = e^{\theta R}$.
\section{Proof of \eqref{eqn:varphi_zt}}\label{app:proof_varphi_zf}
The solution to the homogeneous difference equation \eqref{eqn:varphi_heq} can be solved by employing Z-transform. To this end, it is indispensable to know the initial conditions of the difference equation, and the initial conditions of the $K$-th order linear difference equation are $\varphi(0),\cdots,\varphi(K-1)$. In order to apply Z-transform, we define a new sequence $\varphi_K\left( t \right)$ by shifting the sequence $\varphi\left( {t} \right)$ to the left by $K$ such that $\varphi_K\left( t \right) = \varphi\left( {t + K} \right)$, where $t\ge 0$. Hence, \eqref{eqn:varphi_heq} can be rewritten in terms of $\varphi_K\left( t \right)$ as
\begin{equation}\label{eqn:y_heq}
\varphi_K\left( t \right) - \sum\limits_{\kappa = 1}^K {{a_\kappa}\varphi_K\left( {t - \kappa} \right)} = 0,\,t\ge 0,
\end{equation}
where $\varphi_K\left( -k \right)=\varphi(K-k)$ and $k\in[0,K-1]$.
Applying the translation property of Z-transform \cite[eq.12.4.1]{debnath2010integral} to the both sides of \eqref{eqn:y_heq} gives
\begin{align}\label{eqn:z_trans_varphiK}
&{\cal Z}\left\{ {{\varphi _K}\left( t \right)} \right\}\left( z \right) -\notag\\ &\sum\limits_{\kappa = 1}^K {{a_\kappa}{z^{ - \kappa}}\left( {{\cal Z}\left\{ {{\varphi _K}\left( t \right)} \right\}\left( z \right) + \sum\limits_{l =  - \kappa}^{ - 1} {{\varphi _K}\left( l \right){z^{ - l}}} } \right)}
 = 0,
\end{align}
where ${\cal Z}\left\{ {{\varphi _K}\left( t \right)} \right\}\left( z \right)$ stands for the Z-transform of ${{\varphi _K}\left( t \right)}$. Accordingly, ${\cal Z}\left\{ {{\varphi _K}\left( t \right)} \right\}\left( z \right)$ is simplified as
\begin{equation}\label{eqn:z_trans_sim_varphi}
{\cal Z}\left\{ {{\varphi _K}\left( t \right)} \right\}\left( z \right) = \frac{{\sum\limits_{\kappa = 1}^K {{a_\kappa}{z^{ - \kappa}}\sum\limits_{l = 1}^\kappa {\varphi \left( {K - l} \right){z^l}} } }}{{1 - \sum\limits_{\kappa = 1}^K {{a_\kappa}{z^{ - \kappa}}} }}.
\end{equation}
By applying the inverse Z-transform to \eqref{eqn:z_trans_sim_varphi}, we have
\begin{equation}\label{eqn:inver_z_varphiK}
{\varphi _K}\left( t \right) = \frac{1}{{2\pi \rm i}}\oint\nolimits_c {{\cal Z}\left\{ {{\varphi _K}\left( t \right)} \right\}\left( z \right){z^{t - 1}}dz}.
\end{equation}
Thus, $\varphi(t)$ can be obtained as
\begin{align}\label{eqn:inver_z_varphi_inv}
&\varphi \left( t \right) = \notag\\
&\frac{1}{{2\pi \rm i}}\oint\nolimits_c {\frac{{\sum\limits_{\kappa = 1}^K {\sum\limits_{l = 1}^\kappa {{a_\kappa}\varphi \left( {K - l} \right){z^{t + l - \kappa - 1}}} } }}{{{z^K} - \sum\limits_{\kappa = 1}^K {{a_\kappa}{z^{K - \kappa}}} }}dz} ,\,t\ge K.
\end{align}
With the definition of $\tilde z_1,\cdots,\tilde z_K$, it follows that ${ z^K} - \sum\nolimits_{\kappa = 1}^K {{a_\kappa}{ z^{K - \kappa}}}  = \prod\nolimits _{\kappa=1}^K (z-\tilde z_\kappa)$. By putting the latter into (\ref{eqn:inver_z_varphi_inv}) and using Theorem \ref{the:inv_z_mat}, \eqref{eqn:inver_z_varphi_inv} can be further rewritten as \eqref{eqn:inv_z_varphi_mat}.
\begin{figure*}[!t]
\begin{align}\label{eqn:inv_z_varphi_mat}
\varphi \left( t \right)
&= \frac{{\det \left[ {{{\left( {{\tilde z_i}^{j - 1}} \right)}_{i,1 \le j \le K - 1}},{{\left( {\sum\limits_{\kappa = 1}^K {\sum\limits_{l = 1}^\kappa {{a_\kappa}\varphi \left( {K - l} \right){{\tilde z_i}^{t + l - \kappa - 1}}} } } \right)}_{i,K}}} \right]}}{{\det \tilde{\bf{A}}}}\notag\\
&=\sum\limits_{l = 1}^K {\varphi \left( {K - l} \right)\frac{{\det \overbrace {\left[ {{{\left( {{\tilde z_i}^{j - 1}} \right)}_{i,1 \le j \le K - 1}},{{\left( {\sum\limits_{\kappa = l}^K {{a_\kappa}{{\tilde z_i}^{t + l - \kappa - 1}}} } \right)}_{i,K}}} \right]}^{{\tilde{\bf{B}}_{l}}}}}{{\det \tilde{\bf{A}}}}} .
\end{align}
\hrulefill
\end{figure*}
\section{Proof of the Properties of Effective Capacity}\label{app:pf_prop_ec}
\subsection{Decreasing Monotonicity of $C_e$ w.r.t. $\theta$}
By defining $\ell_t \left( \theta  \right) = -{\ln \mathbb E\left\{ {{e^{ - \theta {S_t}}}} \right\}}/t$, the effective capacity can be expressed as ${C_e} =  \mathop {\lim }\nolimits_{t \to \infty } {{{\ell _t}\left( \theta  \right)}}/{{\theta }}$. More specifically, $\ell_t \left( 0  \right)=0$ and the second derivative of $\ell_t \left( \theta  \right)$ w.r.t. $\theta$ is less than or equal to zero because
\begin{align}\label{eqn:ell_second_der}
\ell_t^{\prime \prime}\left( \theta  \right) &= -\frac{{\mathbb E\left\{ {{S_t}^2{e^{ - \theta {S_t}}}} \right\} \mathbb E\left\{ {{e^{ - \theta {S_t}}}} \right\} - {{\left( {\mathbb E\left\{ {{S_t}{e^{ - \theta {S_t}}}} \right\}} \right)}^2}}}{{{{\left( {\mathbb E\left\{ {{e^{ - \theta {S_t}}}} \right\}} \right)}^2}}t} \notag\\
&\le 0,
\end{align}
where the last step holds by using Cauchy-Schwarz inequality. Hence, we get $\mathop {\lim }\nolimits_{t \to \infty } {{{\ell _t}\left( 0  \right)}} = 0$ and $\mathop {\lim }\nolimits_{t \to \infty } {{{\ell _t}^{\prime \prime}\left( \theta  \right)}} \ge 0$. By using Theorem \ref{the:dec_mon}, the decreasing monotonicity of $C_e$ w.r.t. $\theta$ is thus proved.

\subsection{Bounds of the Effective Capacity}
Owing to the decreasing monotonicity of the effective capacity, the effective capacity is upper bounded as ${C_e} \le \mathop {\lim }\limits_{\theta  \to  0} C_e$, where the upper bound is given by
\begin{equation}\label{eqn:lower_ec_var}
\mathop {\lim }\limits_{\theta  \to  0} C_e = \mathop  {\lim }\limits_{\theta  \to  0} \frac{{\tilde u(\theta )}}{\theta } = \tilde u '(0)=\frac{{\mathbb E\left( {{{\mathcal R(X,\frak S)}}} \right)}}{{\mathbb E\left( X \right)}}.
\end{equation}
On the other hand, we first define the lowest reward as $R_{\hat \kappa,\hat j}$, where $(\hat \kappa,\hat j) = \mathop {\arg }\nolimits_{\left( {\kappa ,j} \right)} \min\nolimits_{q_{\kappa ,j}>0} {R_{\kappa ,j}}$. Hence, the effective capacity satisfies
\begin{equation}\label{eqn:ec_var_eq_rew}
\sum\limits_{\kappa  = 1}^K {\sum\limits_{j = 1}^{{v_\kappa }} {{q_{\kappa ,j}}{e^{ - \theta \left( {{R_{\kappa ,j}} - {R_{\hat \kappa ,\hat j}}} \right)}}{e^{\kappa \theta C_e}}} }  = {e^{\theta {R_{\hat \kappa ,\hat j}}}}.
\end{equation}
From \eqref{eqn:ec_var_eq_rew}, it follows that
\begin{equation}\label{eqn:ec_ine_var}
{q_{\hat \kappa ,\hat j}}{e^{\hat \kappa \theta {C_e}}} \le \sum\limits_{\kappa  = 1}^K {\sum\limits_{j = 1}^{{v_\kappa }} {{q_{\kappa ,j}}{e^{ - \theta \left( {{R_{\kappa ,j}} - {R_{\hat \kappa ,\hat j}}} \right)}}{e^{\kappa \theta {C_e}}}} }  \le {e^{K\theta {C_e}}}.
\end{equation}
Thus $C_e$ is bounded as
\begin{equation}\label{eqn:ec_bound_gen}
\frac{{{R_{\hat \kappa ,\hat j}}}}{K} \le {C_e} \le \frac{{{R_{\hat \kappa ,\hat j}}}}{{\hat \kappa }} - \frac{{\ln {q_{\hat \kappa ,\hat j}}}}{{\hat \kappa \theta }}.
\end{equation}
Combining \eqref{eqn:lower_ec_var} and \eqref{eqn:ec_bound_gen} leads to \eqref{eqn:ec_bound_fin_var}. Consequently, we complete the proof.

\section{Proof of \eqref{eqn:pk_F_G_foxh}}\label{app:pk_hat}
From \eqref{eqn:out_prob_def}, $\hat p_k$ can be rewritten as
\begin{equation}\label{eqn:p_hat_rew}
{{\hat p}_k} = {\rm{Pr}}\left( {\underbrace {\prod\nolimits_{l = 1}^k {{{\left( {1 + {\gamma _l}} \right)}^{\frac{1}{{{\mathpzc R_l}}}}}} }_{{G_k}} < 2} \right)
 = {F_{{G_k}}}\left( 2 \right),
\end{equation}
where $ {F_{{G_k}}}(x)$ refers to the CDF of $G_k$. By using the Mellin transform, the Mellin transform w.r.t. $ {F_{{G_k}}}(x)$ can be obtained by using \cite[eq.8.3.15]{debnath2010integral} as
\begin{equation}\label{eqn:mellin_trans_}
\left\{ {\mathcal M{F_{{G_k}}}} \right\}\left( s \right) = -\frac{1}{s} \left\{ {\mathcal M{f_{{G_k}}}} \right\}\left( s+1 \right),
\end{equation}
where $ {f_{{G_k}}}(x)$ denotes the PDF of $G_k$. The Mellin transform w.r.t. $ {f_{{G_k}}}(x)$ can be expressed as
\begin{align}\label{eqn:mellin_pdf_f}
\left\{ {\mathcal M{f_{{G_k}}}} \right\}\left( s \right)&=\mathbb E \left\{ {G_k}^{s-1}\right\}
\mathop  = \limits^{\left( a \right)}  \prod\limits_{l = 1}^k {\mathbb E \left\{{{\left( {1 + {\gamma _l}} \right)}^{\frac{{s - 1}}{{{\mathpzc R_l}}}}}\right\}  }\notag\\
&= \prod\limits_{l = 1}^k {\int\nolimits_0^\infty  {{{\left( {1 + x} \right)}^{\frac{{s - 1}}{{{\mathpzc R_l}}}}}{f_{{\gamma _l}}}\left( x \right)dx} }.
\end{align}
where step (a) holds because of the independence among fading channels. Note $\gamma_l = \gamma_T \alpha_l$ and $\alpha_l \sim \mathcal G(m_l, \Omega_l/m_l)$, the PDF of $\gamma_l$ is given by
\begin{equation}\label{eqn:gamma_l_pdf}
{f_{{\gamma _l}}}\left( x \right) = {\left( {\frac{{{m_l}}}{{{\gamma _T}{\Omega _l}}}} \right)^{{m_l}}}\frac{{{x^{{m_l} - 1}}{e^{ - \frac{{{m_l}}}{{{\gamma _T}{\Omega _l}}}x}}}}{{\Gamma \left( {{m_l}} \right)}}.
\end{equation}
By substituting \eqref{eqn:gamma_l_pdf} into \eqref{eqn:mellin_pdf_f} yields
\begin{multline}\label{eqn:mellin_trans_pdf_tric}
\left\{ {{\cal M}{f_{{G_k}}}} \right\}\left( s \right) = \prod\limits_{l = 1}^k {{\left( {\frac{{{m_l}}}{{{\gamma _T}{\Omega _l}}}} \right)}^{{m_l}}}\\
 \times \frac{1}{{\Gamma \left( {{m_l}} \right)}}\int\nolimits_0^\infty  {{e^{ - \frac{{{m_l}}}{{{\gamma _T}{\Omega _l}}}x}}{x^{{m_l} - 1}}{{\left( {1 + x} \right)}^{\frac{{s - 1}}{{{\mathpzc R_l}}}}}dx}  \\
 = \prod\limits_{l = 1}^k {{{\left( {\frac{{{m_l}}}{{{\gamma _T}{\Omega _l}}}} \right)}^{{m_l}}}\Psi \left( {{m_l},\frac{{s - 1}}{{{\mathpzc R_l}}} + {m_l} + 1;\frac{{{m_l}}}{{{\gamma _T}{\Omega _l}}}} \right)}.
\end{multline}
Plugging the latter into \eqref{eqn:mellin_trans_} and then applying inverse Mellin transform, $ {F_{{G_k}}}(x)$ can be expressed as
\begin{align}\label{eqn:F_g_mellin_inv}
&{F_{{G_k}}}\left( x \right) = \frac{1}{{2\pi \rm i}}\int\nolimits_{\tilde {\mathcal C}} {\left\{ {{\cal M}{F_{{G_k}}}} \right\}\left( s \right){x^{ - s}}ds}\notag\\
&  =- \prod\limits_{l = 1}^k {{{\left( {\frac{{{m_l}}}{{{\gamma _T}{\Omega _l}}}} \right)}^{{m_l}}}}  \notag\\
&\times\frac{1}{{2\pi \rm i}}\int\nolimits_{\tilde {\mathcal C}} {\frac{1}{s}\prod\limits_{l = 1}^k {\Psi \left( {{m_l},\frac{s}{{{\mathpzc R_l}}} + {m_l} + 1;\frac{{{m_l}}}{{{\gamma _T}{\Omega _l}}}} \right)} {x^{ - s}}ds} ,
\end{align}
where $\tilde {\mathcal C}$ is a Mellin-Barnes contour path. By using \cite[eq.55]{shi2017asymptotic}, ${F_{{G_k}}}\left( x \right)$ can thus be expressed in terms of the generalized Fox's H function \cite{yilmaz2010outage} as \eqref{eqn:F_G_foxh}, shown at the top of the next page. By putting \eqref{eqn:F_G_foxh} into \eqref{eqn:ec_bound_gen}, $\hat p_k$ is consequently obtained as \eqref{eqn:pk_F_G_foxh}.
\begin{figure*}[tbp]
\begin{equation}\label{eqn:F_G_foxh}
 {F_{{G_k}}}\left( x \right) = Y_{1,k + 1}^{k,1}\left[ {\left. {\begin{array}{*{20}{c}}
{\left( {1,1,0,1} \right)}\\
{\left( {1,\frac{1}{{{\mathpzc R_1}}},\frac{{{m_1}}}{{{\gamma _T}{\Omega _1}}},{m_1}} \right), \cdots ,\left( {1,\frac{1}{{{\mathpzc R_k}}},\frac{{{m_k}}}{{{\gamma _T}{\Omega _k}}},{m_k}} \right),\left( {0,1,0,1} \right)}
\end{array}} \right|x\prod\limits_{l = 1}^k {{{\left( {\frac{{{m_l}}}{{{\gamma _T}{\Omega _l}}}} \right)}^{\frac{1}{{{\mathpzc R_l}}}}}} } \right],
\end{equation}
\hrulefill
\end{figure*}
\bibliographystyle{ieeetran}
\bibliography{Asy_ana}
\end{document}